\pdfoutput=1
\RequirePackage{ifpdf}
\ifpdf 
\documentclass[pdftex]{sigma}
\else
\documentclass{sigma}
\fi

\numberwithin{equation}{section}

\newtheorem{Theorem}{Theorem}[section]
\newtheorem{Lemma}[Theorem]{Lemma}
 { \theoremstyle{definition}
\newtheorem{Remark}[Theorem]{Remark} }

\newcommand{\tp}[1]{\,{\vphantom{#1}}^\mathrm{t}\!\,#1}
\newcommand{\CC}{\mathbb{C}}

\newcommand{\ZZ}{\mathbb{Z}}
\newcommand{\calP}{\mathcal{P}}

\newcommand{\bst}{\boldsymbol{t}}
\newcommand{\bsx}{\boldsymbol{x}}
\newcommand{\bsy}{\boldsymbol{y}}

\begin{document}

\allowdisplaybreaks

\newcommand{\arXivNumber}{1609.00882}

\renewcommand{\thefootnote}{}

\renewcommand{\PaperNumber}{009}

\FirstPageHeading

\ShortArticleName{$q$-Dif\/ference Kac--Schwarz Operators in Topological String Theory}

\ArticleName{$\boldsymbol{q}$-Dif\/ference Kac--Schwarz Operators\\ in Topological String Theory\footnote{This paper is a~contribution to the Special Issue on Combinatorics of Moduli Spaces: Integrability, Cohomo\-logy, Quantisation, and Beyond. The full collection is available at \href{http://www.emis.de/journals/SIGMA/moduli-spaces-2016.html}{http://www.emis.de/journals/SIGMA/moduli-spaces-2016.html}}}

\Author{Kanehisa TAKASAKI~$^\dag$ and Toshio NAKATSU~$^\ddag$}

\AuthorNameForHeading{K.~Takasaki and T.~Nakatsu}

\Address{$^\dag$~Department of Mathematics, Kinki University,\\
\hphantom{$^\dag$}~3-4-1 Kowakae, Higashi-Osaka, Osaka 577-8502, Japan}
\EmailD{\href{mailto:takasaki@math.h.kyoto-u.ac.jp}{takasaki@math.h.kyoto-u.ac.jp}}

\Address{$^\ddag$~Institute of Fundamental Sciences, Setsunan University,\\
\hphantom{$^\ddag$}~17-8 Ikeda Nakamachi, Neyagawa, Osaka 572-8508, Japan}
\EmailD{\href{mailto:nakatsu@mpg.setsunan.ac.jp}{nakatsu@mpg.setsunan.ac.jp}}

\ArticleDates{Received September 08, 2016, in f\/inal form February 17, 2017; Published online February 21, 2017}

\Abstract{The perspective of Kac--Schwarz operators is introduced to the authors' previous work on the quantum mirror curves of topological string theory in strip geometry and closed topological vertex. Open string amplitudes on each leg of the web diagram of such geometry can be packed into a multi-variate generating function. This generating function turns out to be a tau function of the KP hierarchy. The tau function has a fermionic expression, from which one f\/inds a vector $|W\rangle$ in the fermionic Fock space that represents a point~$W$ of the Sato Grassmannian. $|W\rangle$ is generated from the vacuum vector $|0\rangle$ by an operator $g$ on the Fock space. $g$~determines an operator $G$ on the space $V = \mathbb{C}((x))$ of Laurent series in which~$W$ is realized as a linear subspace. $G$~generates an admissible basis $\{\Phi_j(x)\}_{j=0}^\infty$ of~$W$. $q$-dif\/ference analogues~$A$,~$B$ of Kac--Schwarz operators are def\/ined with the aid of~$G$. $\Phi_j(x)$'s satisfy the linear equations $A\Phi_j(x) = q^j\Phi_j(x)$, $B\Phi_j(x) = \Phi_{j+1}(x)$. The lowest equation $A\Phi_0(x) = \Phi_0(x)$ reproduces the quantum mirror curve in the authors' previous work.}


\Classification{37K10; 39A13; 81T30}

\renewcommand{\thefootnote}{\arabic{footnote}}
\setcounter{footnote}{0}

\section{Introduction}

It is argued \cite{ADKMV03,DHSV08,DHS09,Hollands09} that a ``quantum mirror curve'' emerges on the B-model (Kodaira--Spencer theory) side of A-model topological string theory when a brane is inserted as a probe. This quantum curve is represented by a~quantum mechanical equation of the form
\begin{gather}
 \hat{H}(\hat{x},\hat{y})\Psi = 0. \label{qmc}
\end{gather}
$\hat{H}(\hat{x},\hat{y})$ is quantization of a function $H(x,y)$ on a 2D phase space,
and the curve
\begin{gather*}
 H(x,y) = 0
\end{gather*}
is called the mirror curve. In the case of local toric Calabi--Yau threefolds, $H(x,y)$ is the function with which the mirror manifold is def\/ined by the equation
\begin{gather*}
 zw = H(x,y).
\end{gather*}
The phase space therein has the logarithmic symplectic form $d\log x\wedge d\log y$, and one can choose~$\hat{x}$ and~$\hat{y}$ as
\begin{gather*}
 \hat{x} = x,\qquad \hat{y} = q^{D},
\end{gather*}
where
\begin{gather*}
 q = e^{-\hbar}, \qquad D = x\frac{d}{dx}.
\end{gather*}
(\ref{qmc}) thereby becomes a $q$-dif\/ference equation for the wave function $\Psi = \Psi(x)$ of the probe.

It is also pointed out \cite{GS11} that quantum curves are closely related to the Eynard--Orantin topological recursion~\cite{EO07}. The notion of topological recursion was introduced to topological string theory through a~matrix model description \cite{DV02,Marino06}. This new computational tool led to the BKMS remodeling conjecture \cite{BKMS07} as well as a new conjectural recursion formula for simple Hurwitz numbers \cite{BM07} (both of which were solved by mathematicians \cite{EMS09,FLZ1310,MZ10}). The mirror curve and a 2-point kernel $B(p,q)$ on this curve are the inputs of topological recursion. The outputs are a set of multi-point kernels $W_{g,n}$, $g \ge 0$, $n \ge 1$, that amount to multi-point resolvents of a matrix model. Amplitudes (or partition functions)
of strings are constructed from these multi-point kernels. Gukov and Su{\l}kowski's observation~\cite{GS11} (based on Eynard and Orantin's interpretation
of the notion of Baker--Akhiezer functions~\cite{EO07}) is that a WKB expansion of the wave function~$\Psi(x)$, too, can be constructed from $W_{g,n}$'s.
Gukov and Su{\l}kowski show, along with many other examples, a partial result on computation of the WKB expansion in the case of $\CC^3$ and the resolved conifold. A complete derivation of the quantum mirror curve for these two cases is achieved by Zhou~\cite{Zhou1207}.

We are interested in deriving the quantum mirror curve from the A-model side, namely, in the topological vertex formalism~\cite{AKMV03}. The wave function $\Psi(x)$ is def\/ined as a generating function of open string amplitudes. Topological vertex of\/fers a combinatorial description of these amplitudes.
If one can compute this generating function explicitly, one should be able to verify the equation~(\ref{qmc}) of the quantum mirror curve by any means.
This is indeed done for~$\CC^3$~\cite{ADKMV03,AKMV03} and the resolved conifold~\cite{HY06,KP06,KP08}. The wave functions for these cases are quantum dilogarithmic functions \cite{FK93,FV93} and $q$-hypergeometric (or basic hypergeometric) functions~\cite{GR-book}, hence satisfy equations of the form~(\ref{qmc}). Some other examples of computations, all falling into ``strip geometry'' \cite{IKP04}, can be found in the literature on dif\/ferent subjects \cite{BTZ11,DGH10,KPW10,Taki10}. A~unif\/ied expression of the quantum mirror curves for general strip geometry is presented in our previous work~\cite{Takasaki13}. Moreover, we extended these results to the simplest example of ``of\/f-strip geometry'' called closed topological vertex~\cite{TN15}.

The aim of this paper is to re-examine our previous results \cite{Takasaki13,TN15} in the perspective of ``Kac--Schwarz operators''. This enables us to highlight a role of the Sato Grassmannian and the KP hierarchy~\cite{SS83,SW85} that was hidden in the previous derivation of quantum mirror curves of topological string theory. Moreover, Kac--Schwarz operators for topological string theory exhibit some new features, which are signif\/icant in the context of integrable hierarchies.

The notion of Kac--Schwarz operators \cite{KS91,Schwarz91} originates in the string equations of 2D quantum gravity \cite{DVV91,Douglas90,FKN91,Moore90}
(see also the recent paper \cite{Schwarz1401} focused on quantum curves). The Kac--Schwarz operators for the $(a,b)$ minimal model of 2D quantum gravity are a pair of ordinary dif\/ferential operators $A$, $B$ of the form
\begin{gather*}
 A = \frac{1}{bz^{b-1}}\frac{d}{dz} + z^a + \text{terms of lower degrees},\qquad B = z^b
\end{gather*}
that satisfy the canonical commutation relation
\begin{gather*}
 [A,B] = 1.
\end{gather*}
They characterize a point $W$ of the Sato Grassmannian as a linear subspace of the space \mbox{$V = \CC((z^{-1}))$} of formal Laurent series by the condition
\begin{gather*}
 AW \subset W,\qquad BW \subset W.
\end{gather*}
The Witten--Kontsevich theory of intersection numbers on a moduli space of curves \cite{Kontsevich92,Witten91}, which amounts to the special case where $(a,b) = (1,2)$, can be treated in the same way.

A variant of this conventional formulation of Kac--Schwarz operators is used to derive quantum spectral curves of various Hurwitz numbers~\cite{Alexandrov1404,ALS1512} and quantum mirror curves of a few cases of topological string theory~\cite{Zhou1512}. The Kac--Schwarz operators therein are dif\/ferential operators of inf\/inite order (in the Hurwitz case) or $q$-dif\/ference operators (in the topological string case), and satisfy the modif\/ied commutation relation
\begin{gather*}
 [A,B] = A \label{[A,B]=B}
\end{gather*}
in the Hurwitz case and
\begin{gather}
 AB = qBA \label{AB=qBA}
\end{gather}
in the topological string case. Moreover, it is pointed out that a particular element $\Psi$ of $W$ (more precisely, the f\/irst element of an admissible basis of~$W$) satisf\/ies the equation
\begin{gather*}
 A\Psi = 0
\end{gather*}
in the Hurwitz case and
\begin{gather*}
 A\Psi = \Psi
\end{gather*}
in the topological string case. These equations are interpreted as quantum spectral/mirror curves.

To capture the quantum mirror curves of general strip geometry and closed topological vertex along these lines, we borrow the idea of {\it generating operator} developed by Alexandrov et al.~\cite{Alexandrov1404,ALS1512}. This is an invertible operator $G$ on $V = \CC((x))$, $x = z^{-1}$, that generates an admissible basis $\{\Phi_j(x)\}_{j=0}^\infty$ of the relevant point $W$ of the Sato Grassmannian. The Kac--Schwarz operators in this case are $q$-dif\/ference operators,
and def\/ined as
\begin{gather}
 A = G\cdot q^{-D}\cdot G^{-1},\qquad B = G\cdot x^{-1}\cdot G^{-1}. \label{ABG-rel}
\end{gather}
The basis functions $\Phi_j(x)$ satisfy the linear equations
\begin{gather}
 A\Phi_j(x) = q^j\Phi_j(x),\qquad B\Phi_j(x) = \Phi_{j+1}(x).
 \label{ABPhij-eq}
\end{gather}
Since $\Phi_0(x)$ turns out to be the generating function of open string amplitudes, the lowest equation $A\Phi_0(x) = \Phi_0(x)$ of (\ref{ABPhij-eq}) can be identif\/ied with the quantum mirror curve of topological string theory. As the algebraic relation (\ref{AB=qBA}) implies, the other equations $A\Phi_j(x) = q^j\Phi_j(x)$ can be derived from the lowest equation by applying $B$ iteratively. Thus $B$ plays the role of a~``ladder operator'' connecting these equations, though we do not know what physical meaning this tower of extra linear equations have.

We derive the generating operator $G$ from a fermionic expression of the KP tau function that is def\/ined as a multi-variate generating function of open string amplitudes. Our computational technique developed in the previous work~\cite{Takasaki13,TN15} yield a fermionic expression of the tau function
almost automatically. This expression contains the vector~$|W\rangle$ of the fermionic Fock space representing~$W$. $|W\rangle$~is generated from
the vacuum vector $|0\rangle$ by an operator~$g$ on the Fock space. In the case of topological string theory of strip geometry and closed topological vertex,
$g$ is made from building blocks that are familiar in the operator formalism of topological vertex \cite{ORV03, BY08}. We can thereby f\/ind an explicit expression of~$G$. This enables us to f\/ind the admissible basis and the $q$-dif\/ference Kac--Schwarz operators as well.

Viewing this machinery on the whole, one may say that the KP hierarchy plays the role of a ``mirror map''. The inputs of this map are open string amplitudes
in the topological vertex formalism. They are manufactured into a KP tau function. This tau function determines a~point~$W$ of the Sato Grassmannian, hence a vector $|W\rangle$ of the fermionic Fock space. From an admissible basis of $W$ and associated Kac--Schwarz operators, one obtains a quantum curve as the output. Such a mapping from the A-model side to the B-model side is commonly called a~mirror map.

This paper is organized as follows. In Section~\ref{section2}, the notion of generating operators and Kac--Schwarz operators are introduced and illustrated
for topological string theory on~$\CC^3$ and the resolved conifold. These two cases are presented as a prototype of the subsequent case studies. Relation to ``hypergeometric tau functions'' is brief\/ly mentioned in the end of this section. Section~\ref{section3} deals the case of general strip geometry. A~multi-variate generating function of open string amplitudes is constructed for each external line, or ``leg'', of the web diagram, and shown to be a KP tau function with a fermionic expression. Actually, the tau functions on the vertical legs and those on the non-vertical legs turn out to have distinct structures. This dif\/ference is ref\/lected in the types of functions forming the admissible bases. The generating operator~$G$ and the $q$-dif\/ference operators~$A$,~$B$ are computed explicitly. Results of the previous work~\cite{Takasaki13} are reproduced successfully. Section~\ref{section4} presents similar results in the case of closed topological vertex. The relevant tau functions are shown to have a slightly more complicated structure. This characteristic is inherited by the functions forming the associated admissible basis. Nevertheless, computations are mostly parallel to the case of strip geometry. Actually, the previous work~\cite{TN15} can be reproduced in a more systematic way in the language of the operators $G$, $A$, $B$.

\section[$q$-dif\/ference Kac--Schwarz operators]{$\boldsymbol{q}$-dif\/ference Kac--Schwarz operators}\label{section2}

\subsection{Generating operator of admissible basis}

A point $W$ of the big cell of the Sato Grassmannian is realized by a linear subspace
\begin{gather*}
 W = \operatorname{Span}\{\Phi_j(x)\}_{j=0}^\infty
\end{gather*}
of the space $V = \CC((x))$ of formal Laurent series with an admissible basis
\begin{gather*}
 \Phi_j(x) = \sum_{i=-j}^\infty \phi_{ij}x^i,\qquad \phi_{-j,j} \not= 0,\qquad j = 0,1,\ldots.
\end{gather*}
For convenience, we allow the leading coef\/f\/icients $\phi_{-j,j}$ to be unnormalized, though they can be normalized to~$1$. The variable~$x$ amounts to the inverse of the usual spectral parameter~$z$ of the KP hierarchy:
\begin{gather*}
 x = z^{-1}.
\end{gather*}

If the admissible basis is generated by a single invertible $q$-dif\/ference operator $G = G(x,q^D)$ as
\begin{gather}
 \Phi_j(x) = Gx^{-j},\qquad j = 0,1,\dots, \label{PhijG-rel}
\end{gather}
the basis elements satisfy the linear $q$-dif\/ference equations~(\ref{ABPhij-eq}) with respect to the operators~$A$,~$B$ def\/ined by~(\ref{ABG-rel}). In particular, the f\/irst basis element $\Phi_0(x)$ satisf\/ies the equation
\begin{gather}
 A\Phi_0(x) = \Phi_0(x), \label{APhi0-eq}
\end{gather}
which we identify with the quantum mirror curve of topological string theory.

(\ref{ABG-rel}) are $q$-dif\/ference analogues of the Kac--Schwarz operators for Hurwitz numbers of various types \cite{Alexandrov1404,ALS1512,Zhou1512}. In those cases, $G$ is a dif\/ferential operator of inf\/inite order, $G = G(x,D)$. The Kac--Schwarz operators $A$, $B$ are def\/ined as
\begin{gather*}
 A = - G\cdot D\cdot G^{-1},\qquad B = G\cdot x^{-1}\cdot G^{-1}
\end{gather*}
and satisfy the twisted canonical commutation relation (\ref{AB=qBA}). The admissible basis generated by~$G$, in turn, satisf\/ies the linear dif\/ferential equations
\begin{gather*}
 A\Phi_j(x) = j\Phi_j(x),\qquad B\Phi_j(x) = \Phi_{j+1}(x).
\end{gather*}
The lowest equation
\begin{gather*}
 A\Phi_0(x) = 0
\end{gather*}
gives the quantum spectral curve for Hurwitz numbers.

$q$-dif\/ference Kac--Schwarz operators are presented by Zhou~\cite{Zhou1512} for open string amplitudes of~$\CC^3$ and the resolved conifold with a general framing parameter~$f$. We review these results (in the case where $f = 0$) as prototypes of our subsequent consideration. Unlike Zhou's formulation, we make full use of the theory of Schur functions~\cite{Mac-book} and the fermionic formalism of integrable hierarchies~\cite{MJD-book}.

\begin{figure}[t]\centering
\includegraphics[scale=0.7]{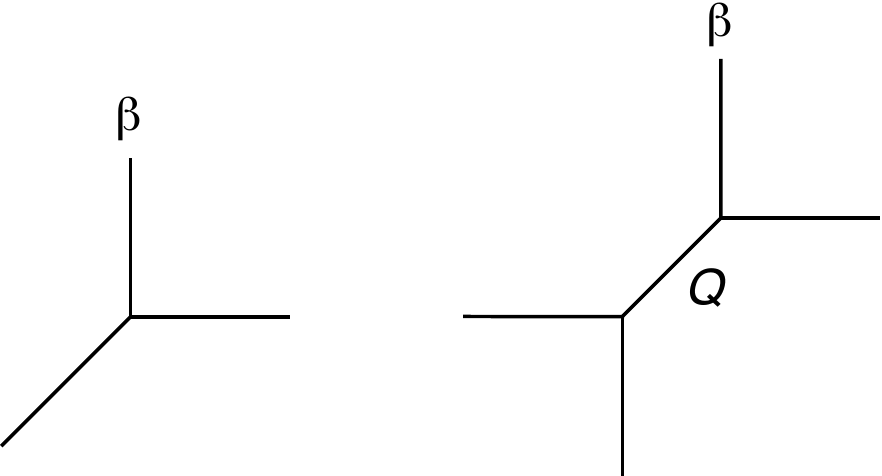}
\caption{Web diagram of $\CC^3$ (left) and resolved conifold (right).}\label{fig1}
\end{figure}

\subsection[Example: $\CC^3$]{Example: $\boldsymbol{\CC^3}$}

The open string amplitudes $Z_\beta$, $\beta = (\beta_1,\beta_2,\ldots) \in \calP$ (the set of all partitions), of topological string theory on~$\CC^3$ in the presence of a stack of branes (Fig.~\ref{fig1}, left) are given by special values of the inf\/inite-variate Schur functions $s_\lambda$ as
\begin{gather*}
 Z_\beta= s_{\tp{\beta}}\big(q^{-\rho}\big),\qquad q^{-\rho} = \big(q^{1/2},q^{3/2},\ldots,q^{i-1/2},\ldots\big),
\end{gather*}
where $\tp{\beta}$ denotes the conjugate partition of $\beta$. $Z_\beta$ is nothing but the special case $C_{\varnothing\varnothing\beta}$ of the topological vertex.

We use the Schur functions $s_{\tp{\beta}}(\bsx)$ of auxiliary variables $\bsx = (x_1,x_2,\ldots)$ to def\/ine a generating function $\tau(\bsx)$ of~$Z_\beta$'s:
\begin{gather*}
 \tau(\bsx) = \sum_{\beta\in\calP}s_{\tp{\beta}}(\bsx)Z_\beta.
\end{gather*}
This sum can be readily computed by the Cauchy identity of Schur functions:
\begin{gather*}
 \tau(\bsx) = \sum_{\beta\in\calP}s_{\tp{\beta}}(\bsx)s_{\tp{\beta}}\big(q^{-\rho}\big) = \prod_{i,j=1}^\infty\big(1 - q^{i-1/2}x_j\big)^{-1}.
\end{gather*}
By the so called Miwa transformation
\begin{gather}
 t_k = \frac{1}{k}\sum_{i=1}^\infty x_i^k \label{tx-rel}
\end{gather}
to the time variables $\bst = (t_1,t_2,\ldots)$ of the KP hierarchy, this inf\/inite product turns into an exponential function:
\begin{gather*}
 \tau(\bsx) = \exp\left(\sum_{k=1}^\infty\frac{q^{k/2}}{1-q^k}t_k\right).
\end{gather*}
This is an exponential tau function of the KP hierarchy.

To f\/ind an associated generating operator $G$, it is convenient to start from the fermionic expression
\begin{gather}
 \tau(\bsx) = \big\langle 0|\Gamma_{+}(\bsx)\Gamma_{-}\big(q^{-\rho}\big)|0\big\rangle \label{tau-CP3}
\end{gather}
of $\tau(\bsx)$. $\langle 0|$ and $|0\rangle$ are the vacuum vectors of the fermionic Fock spaces on which the creation-annihilation operators~$\psi_n$,~$\psi^*_n$, $n \in \ZZ$, of a 2D charged free fermions act. $\Gamma_{\pm}(\bsx)$ are multi-variate extensions
\begin{gather*}
 \Gamma_{\pm}(\bsx) = \prod_{i=1}^\infty\Gamma_{\pm}(x_i)
\end{gather*}
of the vertex operators \cite{ORV03, BY08}
\begin{gather*}
 \Gamma_{\pm}(z) = \exp\left(\sum_{k=1}^\infty\frac{z^k}{k}J_{\pm k}\right),
\end{gather*}
where $J_k$'s are basis elements of the $U(1)$ current algebra:
\begin{gather*}
 J_k = \sum_{n\in\ZZ}{:}\psi_{k-n}\psi^*_n{:}.
\end{gather*}
Note that the Schur functions can be expressed as
\begin{gather*}
 s_\lambda(\bsx) = \langle 0|\Gamma_{+}(\bsx)|\lambda\rangle = \langle\lambda|\Gamma_{-}(\bsx)|0\rangle.
\end{gather*}
Assembling these building blocks, one obtains the fermionic expression~(\ref{tau-CP3}) of~$\tau(\bsx)$. Moreover, by the Miwa transformation~(\ref{tx-rel}), $\Gamma_{+}(\bsx)$ turns into the generator
\begin{gather*}
 \gamma_{+}(\bst) = \exp\left(\sum_{k=1}^\infty t_kJ_k\right)
\end{gather*}
of time evolutions of the KP hierarchy. (\ref{tau-CP3}) thereby becomes the standard expression
\begin{gather*}
 \tau(\bsx) = \big\langle 0|\gamma_{+}(\bst)\Gamma_{-}\big(q^{-\rho}\big)|0\big\rangle
\end{gather*}
of a KP tau function. The vector
\begin{gather*}
 |W\rangle = \Gamma_{-}\big(q^{-\rho}\big)|0\rangle
\end{gather*}
in the Fock space represents the associated point~$W$ of the Sato Grassmannian. This vector is generated from $|0\rangle$ by the operator~$g=\Gamma_{-}(q^{-\rho})$.

The generating operator $G$ can be obtained by translating the operator~$g$ on the fermionic Fock space to an operator on the space~$V$ of Laurent series.
The rule of translation is based on the following well known correspondences \cite{MJD-book,SS83,SW85} (see also Alexandrov's detailed account~\cite{Alexandrov1404}):
\begin{itemize}\itemsep=0pt
\item[1.] The correspondence
\begin{gather*}
 \hat{A} = \sum_{i,j\in\ZZ}a_{ij}{:}\psi_{-i}\psi^*_j{:}
 \longleftrightarrow
 A = (a_{i,j})_{i,j\in\ZZ}
\end{gather*}
between fermion bilinears and $\ZZ\times\ZZ$ matrices
based on the commutation relations
\begin{gather*}
 [\hat{A},\hat{B}] = \widehat{[A,B]} + \text{$c$-number term}.
\end{gather*}
\item[2.] The correspondence
\begin{gather*}
 D \longleftrightarrow \Delta = (i\delta_{ij})_{i,j\in\ZZ},\\
 q^D \longleftrightarrow q^\Delta = (q^i\delta_{ij})_{i,j\in\ZZ},\\
 x^{-k} \longleftrightarrow \Lambda^k = (\delta_{i+k,j})_{i,j\in\ZZ}
\end{gather*}
between dif\/ferential/dif\/ference operators and $\ZZ\times\ZZ$ matrices based on the relations
\begin{gather*}
 D\sum_{i\in\ZZ}\phi_ix^i = \sum_{i\in\ZZ}i\phi_ix^i,\\
 q^D\sum_{i\in\ZZ}\phi_ix^i = \sum_{i\in\ZZ}q^i\phi_ix^i,\\
 x^{-k}\sum_{i\in\ZZ}\phi_ix^i = \sum_{i\in\ZZ}\phi_{i+k}x^i.
\end{gather*}
\end{itemize}

By this rule, the operator
\begin{gather*}
 g = \Gamma_{-}(q^{-\rho}) = \exp\left(\sum_{i,k=1}^\infty\frac{q^{k(i-1/2)}}{k}J_{-k}\right)
\end{gather*}
on the fermionic Fock space turns into the multiplication operator
\begin{gather}
 G = \exp\left(\sum_{i,k=1}^\infty\frac{q^{k(i-1/2)}}{k}x^k\right) = \prod_{i=1}^\infty\big(1 - q^{i-1/2}x\big)^{-1} \label{G-CP3}
\end{gather}
on $V$. The multiplier is a quantum dilogarithmic function. The admissible basis (\ref{PhijG-rel}) consists of the functions
\begin{gather}
 \Phi_j(x) = x^{-j}\prod_{i=1}^\infty\big(1 - q^{i-1/2}x\big)^{-1} = \sum_{k=0}^\infty \frac{q^{k/2}}{(q;q)_k}x^{k-j}, \label{Phij-CP3}
\end{gather}
where we have used the $q$-Pochhammer symbol
\begin{gather*}
 (a;q)_k = \begin{cases}
 (1 - a)(1 - aq)\cdots \big(1 - aq^{k-1}\big) &\text{if $k > 0$},\\
 1 & \text{if $k = 0$}.
 \end{cases}
\end{gather*}

It is straightforward to compute the $q$-dif\/ference Kac--Schwarz operators (\ref{ABG-rel}) from (\ref{G-CP3}). The operator $A = G\cdot q^{-D}\cdot G^{-1}$ can be computed with the aid of the operator identity
\begin{gather}
 q^Df(x)q^{-D} = f(qx) \label{q^D-conj}
\end{gather}
as
\begin{gather}
 A = q^{-D}\prod_{i=1}^\infty\big(1 - q^{i+1/2}x\big)^{-1} \cdot\prod_{i=1}^\infty\big(1 - q^{i-1/2}x\big) = q^{-D}\big(1 - q^{1/2}x\big). \label{A-CP3}
\end{gather}
The operator $B = G\cdot x^{-1}\cdot G^{-1}$ becomes a trivial one:
\begin{gather*}
 B = x^{-1}.
\end{gather*}

The $q$-dif\/ference equation (\ref{APhi0-eq}) for $\Phi_0(x)$ reads
\begin{gather*}
 q^{-D}\big(1 - q^{1/2}x\big)\Phi_0(x) = \Phi_0(x)
\end{gather*}
or, equivalently,
\begin{gather*}
 \big(1 - q^{1/2}x - q^D\big)\Phi_0(x) = 0,
\end{gather*}
which agrees with the well known equation of the quantum mirror curve of $\CC^3$ \cite{ADKMV03,AKMV03}.

\subsection{Example: Resolved conifold}

The foregoing construction of the operators $G$, $A$, $B$ can be extended to the resolved conifold. The unnormalized open string amplitudes~$Z_\beta$ with a~stack of branes on one of the external lines of the web diagram (Fig.~\ref{fig1}, right) can be expressed as
\begin{gather*}
 Z_\beta = s_{\tp{\beta}}\big(q^{-\rho}\big) \prod_{i,j=1}^\infty\big(1 - Qq^{-\beta_i+i+j-1}\big),
\end{gather*}
where $Q$ is the K\"ahler parameter of the internal line of the web diagram. We consider the generating function
\begin{gather*}
 \tau(\bsx) = \sum_{\beta\in\calP}s_{\tp{\beta}}(\bsx)Z_\beta/Z_\varnothing = \sum_{\beta\in\calP}s_{\tp{\beta}}(\bsx)s_{\tp{\beta}}\big(q^{-\rho}\big)
 \prod_{i,j=1}^\infty\dfrac{1 - Qq^{-\beta_i+i+j-1}}{1 - Qq^{i+j-1}}
\end{gather*}
of the normalized amplitudes $Z_\beta/Z_\varnothing$.

To derive a fermionic expression of $\tau(\bsx)$, we recall the following formula, $\lambda = (\lambda_1,\lambda_2,\ldots) \in \calP$, from the previous work~\cite{TN15}:
\begin{gather}
 \prod_{i,j=1}^\infty\dfrac{1 - Qq^{-\lambda_i+i+j-1}}{1 - Qq^{i+j-1}} = \left\langle\tp{\lambda}\,|\, \exp\left(\sum_{k=1}^\infty\frac{Q^k}{k(1-q^k)}V^{(k)}_0\right) |\,\tp{\lambda}\right\rangle. \label{V(k)_0}
\end{gather}
$\langle\lambda|$ and $|\lambda\rangle$ denote the excited states
\begin{gather*}
 \langle \lambda| = \langle -\infty\,|\,\cdots\psi^*_{\lambda_i-i+1} \cdots\psi^*_{\lambda_2-1}\psi^*_{\lambda_1},\qquad
 |\lambda\rangle = \psi_{-\lambda_1}\psi_{-\lambda_2+1}\cdots \psi_{-\lambda_i+i-1}\cdots\,|\,-\infty\rangle
\end{gather*}
in the charge-$0$ sector of the Fock spaces. The ground states $\langle\varnothing|,|\varnothing\rangle$ labeled by the zero partition are nothing but
the vacuum vectors $\langle 0|,|0\rangle$. $V^{(k)}_0$'s are the $m = 0$ part of the basis
\begin{gather*}
 V^{(k)}_m = q^{-km/2}\sum_{n\in\ZZ}q^{kn}{:}\psi_{m-n}\psi^*_{n}{:},\qquad k,m\in\ZZ,
\end{gather*}
of the quantum torus algebra \cite{NT07,OP03}. Since $V^{(k)}_0$'s do not mix the excited states, the fermionic expression (\ref{V(k)_0}) of the $Q$-dependent factor in the def\/inition of $\tau(\bsx)$ can be merged to the fermionic expression
\begin{gather*}
 s_{\tp{\beta}}(q^{-\rho}) = \big\langle\tp{\beta}\,|\,\Gamma_{-}(q^{-\rho})\,|\,0\big\rangle
\end{gather*}
of $s_{\tp{\beta}}(q^{-\rho})$ as
\begin{gather*}
 \prod_{i,j=1}^\infty\dfrac{1 - Qq^{-\beta_i+i+j-1}}{1 - Qq^{i+j-1}}
 s_{\tp{\beta}}\big(q^{-\rho}\big) = \left\langle\tp{\beta}\,|\, \exp\left(\sum_{k=1}^\infty\frac{Q^k}{k(1-q^k)}V^{(k)}_0\right)
 \Gamma_{-}\big(q^{-\rho}\big)\,|\,0\right\rangle.
\end{gather*}
Thus $\tau(\bsx)$ can be expressed in the fermionic form
\begin{gather}
 \tau(\bsx) = \left\langle 0\,|\,\Gamma_{+}(\bsx) \exp\left(\sum_{k=1}^\infty\frac{Q^k}{k(1-q^k)}V^{(k)}_0\right)
 \Gamma_{-}\big(q^{-\rho}\big)\,|\,0\right\rangle. \label{tau-conifold}
\end{gather}
The associated point $W$ of the Sato Grassmannian is determined by the vector
\begin{gather*}
 |W\rangle = \exp\left(\sum_{k=1}^\infty\frac{Q^k}{k(1-q^k)}V^{(k)}_0\right) \Gamma_{-}\big(q^{-\rho}\big)|0\rangle
\end{gather*}
in the fermionic Fock space.

The generating operator $G$ of an admissible basis of $W$ can be obtained by translating the generating operator
\begin{gather*}
 g = \exp\left(\sum_{k=1}^\infty\frac{Q^k}{k(1-q^k)}V^{(k)}_0\right) \Gamma_{-}\big(q^{-\rho}\big)
\end{gather*}
to an operator on $V$. We have seen that $\Gamma_{-}(q^{-\rho})$ amounts to the multiplicative operator~(\ref{G-CP3}). As regard the f\/irst factor, since
\begin{gather*}
 V^{(k)}_0 = \sum_{n\in\ZZ}q^{kn}{:}\psi_{-n}\psi^*_n{:}
\end{gather*}
corresponds to the $q$-shift operator $q^{kD}$, one can compute its avatar on $V$ as
\begin{gather*}
 \exp\left(\sum_{k=1}^\infty\frac{Q^k}{k(1-q^k)}q^{kD}\right) = \exp\left(\sum_{i,k=1}^\infty\frac{(Qq^{i-1+D})^k}{k}\right)
 = \prod_{i=1}^\infty \big(1 - Qq^{i-1+D}\big)^{-1}.
\end{gather*}
$G$ is a product of these two operators:
\begin{gather}
 G = \prod_{i=1}^\infty \big(1 - Qq^{i-1+D}\big)^{-1} \cdot\prod_{i=1}^\infty \big(1 - q^{i-1/2}x\big)^{-1}. \label{G-conifold}
\end{gather}

To compute the admissible basis (\ref{PhijG-rel}), one can partly make use of the results for~$\CC^3$. Namely, applying $G$ to $x^{-j}$'s amount to applying
$\prod\limits_{i=1}^\infty (1 - Qq^{i-1+D})^{-1}$ to the admissible basis~(\ref{Phij-CP3}) of~$\CC^3$. Thus the admissible basis in this case can be computed as
\begin{gather*}
 \Phi_j(x) = \prod_{i=1}^\infty \big(1 - Qq^{i-1+D}\big)^{-1} \left(x^{-j}\prod_{i=1}^\infty\big(1 - q^{i-1/2}x\big)^{-1}\right) \\
\hphantom{\Phi_j(x)}{} = \prod_{i=1}^\infty \big(1 - Qq^{i-1+D}\big)^{-1}
 \left(\sum_{k=0}^\infty\frac{q^{k/2}}{(q;q)_k}x^{k-j}\right) = \sum_{k=0}^\infty\frac{q^{k/2}}{(q;q)_k}x^{k-j}
 \prod_{i=1}^\infty\big(1 - Qq^{i-1+k-j}\big)^{-1}.
\end{gather*}
Rewriting this expression slightly, one f\/inds that
\begin{gather}
 \Phi_j(x) = \frac{1}{\prod\limits_{i=1}^\infty(1- q^{i-j-1})} \sum_{k=0}^\infty\frac{(Qq^{-j};q)_kq^{k/2}}{(q;q)_k}x^{k-j}. \label{Phij-conifold}
\end{gather}
Thus the essential part of $\Phi_j(x)$'s turn out to be $q$-hypergeometric functions.

One can similarly cut unnecessary computation for the $q$-dif\/ference Kac--Schwarz opera\-tors~(\ref{ABG-rel}). Namely, one can use~(\ref{A-CP3}) to rewrite
$A = G\cdot q^{-D}\cdot G^{-1}$ as
\begin{gather*}
 A = \prod_{i=1}^\infty\big(1 - Qq^{i-1+D}\big)^{-1} \cdot q^{-D}\big(1 - q^{1/2}x\big) \cdot \prod_{i=1}^\infty\big(1 - Qq^{i-1+D}\big)\\
\hphantom{A=}{} = q^{-D} - q^{-D}q^{1/2}\prod_{i=1}^\infty\big(1 - Qq^{i-1+D}\big)^{-1} \cdot x\cdot \prod_{i=1}^\infty\big(1 - Qq^{i-1+D}\big).
\end{gather*}
By the operator identity
\begin{gather}
 x^{-1}f(D)x = f(D+1), \label{x-conj}
\end{gather}
the main part of the second term on the right side can be computed as
\begin{gather*}
 \prod_{i=1}^\infty\big(1 - Qq^{i-1+D}\big)^{-1} \cdot x\cdot \prod_{i=1}^\infty\big(1 - Qq^{i-1+D}\big)
= x\prod_{i=1}^\infty\big(1 - Qq^{i+D}\big)^{-1} \cdot\prod_{i=1}^\infty\big(1 - Qq^{i-1+D}\big) \\
\hphantom{\prod_{i=1}^\infty\big(1 - Qq^{i-1+D}\big)^{-1} \cdot x\cdot \prod_{i=1}^\infty\big(1 - Qq^{i-1+D}\big)}{} = x\big(1 - Qq^D\big).
\end{gather*}
Thus $A$ can be expressed as
\begin{gather}
 A = q^{-D}\big(1 - q^{1/2}x\big(1 - Qq^D\big)\big). \label{A-conifold}
\end{gather}
Moreover, this computation shows that $B$ can be obtained as an inverse of $x(1 - Qq^D)$:
\begin{gather}
 B = \big(1 - Qq^D\big)^{-1}x^{-1}. \label{B-conifold}
\end{gather}

The $q$-dif\/ference equation (\ref{APhi0-eq}) for $\Phi_0(x)$ reads
\begin{gather*}
 q^{-D}\big(1 - q^{1/2}x\big(1 - Qq^D\big)\big)\Phi_0(x) = \Phi_0(x).
\end{gather*}
This equation can be converted to the equation
\begin{gather}
 \big(1 - q^{1/2}x\big)\Phi_0(x) = \big(1 - Qq^{1/2}x\big)\Phi_0(qx), \label{Phi0-eq-conifold}
\end{gather}
known in the literature \cite{GS11,HY06,KP06,KP08,Zhou1207}.

\subsection{Relation to hypergeometric tau functions}

(\ref{tau-CP3}) and (\ref{tau-conifold}) belong to a class of KP tau functions that can be expressed in the $\bsx$-variables as
\begin{gather}
 \tau(\bsx) = \big\langle 0\,|\,\Gamma_{+}(\bsx)h\Gamma_{-}\big(q^{-\rho}\big)\,|\,0\big\rangle, \label{tauh-gen}
\end{gather}
where $h$ is an operator of the form
\begin{gather}
 h = \exp\left(\sum_{n\in\ZZ}\log h_n{:}\psi_{-n}\psi^*_n{:}\right).
 \label{h-gen}
\end{gather}
In the correspondence with $\ZZ\times\ZZ$ matrices, $h$ amounts to a diagonal matrix, and the parame\-ters~$h_n$ are nothing but the diagonal matrix elements.
The generating operator~$G$ on~$V$ corresponds to the generating operator $g = h\Gamma_{-}(q^{-\rho})$ on the fermionic Fock space.

These KP tau functions can be derived from tau functions of the 2D Toda hierarchy of the form
\begin{gather}
 \tau(s,\bst,\tilde{\bst}) = \big\langle s\,|\,\gamma_{+}(\bst)h\gamma_{-}(\tilde{\bst})\,|\,s\big\rangle, \label{tauh2D-gen}
\end{gather}
where
\begin{gather*}
 \gamma_{-}(\tilde{\bst}) = \exp\left(\sum_{k=1}^\infty\tilde{t}_kJ_{-k}\right),
\end{gather*}
by specialization to
\begin{gather*}
 s = 0, \qquad t_k = \frac{1}{k}\sum_{i=1}^\infty x_i^k,\qquad \tilde{t}_k = \frac{q^{k/2}}{k(1 - q^k)}.
\end{gather*}
Toda tau functions of the form (\ref{tauh2D-gen}), called ``hypergeometric tau functions'', were introduced by Orlov and Scherbin in their study on multi-variate hypergeometric functions \cite{OS01a,OS01b}. As pointed out therein (see also the work of Kharchev, Marshakov, Mironov and Morozov~\cite{KMM95} in a~dif\/ferent context), $h$ is diagonal with respect to the basis $\{|\lambda\rangle\}_{\lambda\in\calP}$ of the charge-$0$ sector, and the diagonal matrix elements have the contents-product form
\begin{gather*}
 \langle\lambda\,|\,h\,|\,\lambda\rangle = \prod_{(i,j)\in\lambda}r_{j-i+1},
\end{gather*}
where
\begin{gather*}
 r_n = \frac{h_n}{h_{n-1}}.
\end{gather*}
Since the work Okounkov \cite{Okounkov00}, hypergeometric tau functions have played a signif\/icant role as generating functions of Hurwitz numbers
and many other combinatorial notions \cite{AMMN14,GPH1408,GPH1405,Harnad1410,Harnad1504,HO1407,MSS13,Takasaki10,Zabrodin12} (see also Harnad's comprehensive review \cite{Harnad1504review}).

In the case of the resolved conifold, see (\ref{tau-conifold}), the operator $h$ can be expressed as
\begin{gather*}
 h = \exp\left(\sum_{k=1}^\infty\frac{Q^k}{k(1-q^k)}V^{(k)}_0\right).
 \label{h-conifold}
\end{gather*}
Since
\begin{gather*}
 \sum_{k=1}^\infty\frac{Q^k}{k(1-q^k)}V^{(k)}_0 = \sum_{n\in\ZZ}\sum_{k=1}^\infty\frac{Q^kq^{kn}}{k(1-q^k)} {:}\psi_{-n}\psi^*_n{:},
\end{gather*}
the parameters $h_n$, $r_n$ can be computed as
\begin{gather*}
 h_n = \exp\left(\sum_{k=1}^\infty\frac{Q^kq^{kn}}{k(1-q^k)}\right) = \prod_{i=1}^\infty\big(1 - Qq^{i-1+n}\big)^{-1}
\end{gather*}
and
\begin{gather*}
 r_n = 1 - Qq^{n-1}.
\end{gather*}
This belongs to a case considered by Orlov and Scherbin in the context of $q$-hypergeometric functions. As we show in the next section, tau functions in more general strip geometry, too, have a similar interpretation.

\section{Strip geometry}\label{section3}

\subsection{Generalities}

The toric diagram of strip geometry (Fig.~\ref{fig2}) is a strip of height $1$ divided to triangles of area~$1/2$. If the toric diagram consists of $N$ triangles, the associated web diagram is acyclic and has $N$ vertices, $N-1$ internal lines and $N+2$ external lines (called ``legs'' for short).
The legs other than the leftmost and rightmost ones are vertical. Let $Q_1,\ldots,Q_{N-1}$ denote the K\"ahler parameters assigned to the internal lines.

\begin{figure}[t] \centering
\includegraphics[scale=0.7]{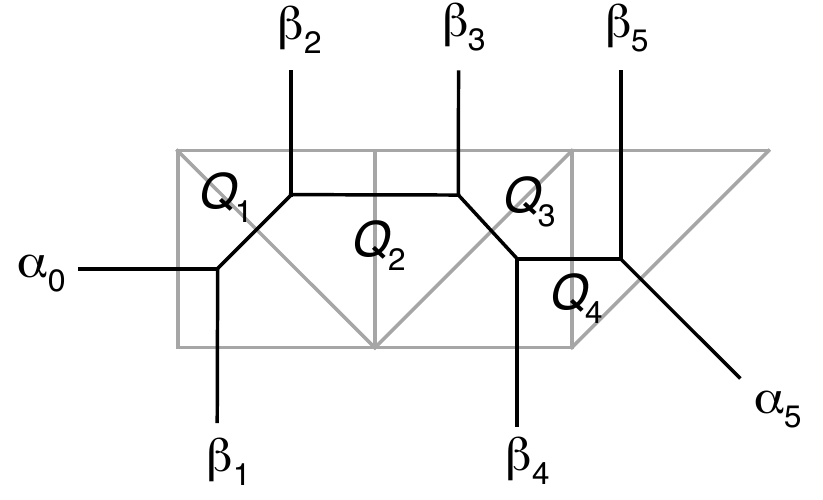}
\caption{Web diagram (solid) and toric diagram (gray) of strip geometry.}\label{fig2}
\end{figure}

Let $Z^{\alpha_0\alpha_N}_{\beta_1\cdots\beta_N}$ denote the open string amplitudes of this geometry. $\alpha_0$, $\alpha_N$ are partitions on the leftmost and rightmost legs, and $\beta_1,\ldots,\beta_N$ are partitions on the vertical legs. The method of topological vertex is known to work particularly well for these amplitudes.

In the case where $\alpha_0 = \alpha_N = \varnothing$, these amplitudes are computed by Iqbal and Kashani-Poor~\cite{IKP04} in a closed form. To state their result, we def\/ine the sign $\sigma_n = \pm 1$ of the $n$-th vertex in accordance with the direction of the vertical leg emanating from the vertex:
\begin{itemize}\itemsep=0pt
\item[(a)] $\sigma_n = +1$ if the vertical leg points up,
\item[(b)] $\sigma_n = -1$ if the vertical leg points down.
\end{itemize}
We further introduce the auxiliary notations
\begin{gather*}
 \beta^{(n)} = \begin{cases}
 \beta_n &\text{if $\sigma_n = +1$},\\
 \tp{\beta}_n &\text{if $\sigma_n = -1$},
 \end{cases}
 \qquad Q_{mn} = Q_mQ_{m+1}\cdots Q_{n-1}.
\end{gather*}
The result of Iqbal and Kashani-Poor thereby reads{\samepage
\begin{gather}
Z^{\varnothing\varnothing}_{\beta_1\cdots\beta_N} = s_{\tp{\beta}_1}\big(q^{-\rho}\big)\cdots s_{\tp{\beta}_N}\big(q^{-\rho}\big)
\!\prod_{1\leq m<n\leq N}\prod_{i,j=1}^\infty\!
 \big(1 - Q_{mn}q^{-\tp{\beta}^{(m)}_i-\beta^{(n)}_j+i+j-1}\big)^{-\sigma_m\sigma_n}, \!\!\!\label{Z-strip-IKP}
\end{gather}
where $\tp{\beta}^{(n)}_i$ and $\beta^{(n)}_i$ denote the $i$-th parts of $\tp{\beta}^{(n)}$ and $\beta^{(n)}$.}

The amplitudes not restricted to $\alpha_0 = \alpha_N = \varnothing$ are known to have the following fermionic expression~\cite{EK03,Nagao09,Sulkowski09}:
\begin{gather}
 Z^{\alpha_0\alpha_N}_{\beta_1\cdots\beta_N}= q^{(1-\sigma_1)\kappa(\alpha_0)/4}
 q^{(1+\sigma_N)\kappa(\alpha_N)/4} s_{\tp{\beta}_1}\big(q^{-\rho}\big)\cdots s_{\tp{\beta}_N}\big(q^{-\rho}\big) \notag\\
\hphantom{Z^{\alpha_0\alpha_N}_{\beta_1\cdots\beta_N}=}{} \times
\big \langle\tp{\alpha}_0\,| \,\Gamma^{\sigma_1}_{-}\big(q^{-\beta^{(1)}-\rho}\big)
 \Gamma^{\sigma_1}_{+}\big(q^{-\tp{\beta}^{(1)}-\rho}\big)(\sigma_1Q_1\sigma_2)^{L_0} \cdots \notag\\
\hphantom{Z^{\alpha_0\alpha_N}_{\beta_1\cdots\beta_N}=}{} \times
 \Gamma^{\sigma_{N-1}}_{-}\big(q^{-\beta^{(N-1)}-\rho}\big)
 \Gamma^{\sigma_{N-1}}_{+}\big(q^{-\tp{\beta}^{(N-1)}-\rho}\big)
 \big(\sigma_{N-1}Q_{N-1}\sigma_N\big)^{L_0} \notag\\
\hphantom{Z^{\alpha_0\alpha_N}_{\beta_1\cdots\beta_N}=}{}\times
 \Gamma^{\sigma_N}_{-}\big(q^{-\beta^{(N)}-\rho}\big)
 \Gamma^{\sigma_N}_{+}\big(q^{-\tp{\beta}^{(N)}-\rho}\big)\,|\,\alpha_N\big\rangle, \label{Z-strip-gen}
\end{gather}
where we have used the common notations
\begin{gather*}
 \kappa(\lambda) = \sum_{i=1}^\infty\lambda_i(\lambda_i - 2i + 1),\qquad L_0 = \sum_{n\in\ZZ}n{:}\psi_{-n}\psi^*_n{:},
\end{gather*}
and $\Gamma^\sigma_{\pm}(\bsx)$ denote either $\Gamma_{\pm}(\bsx)$ or the vertex operators
\begin{gather*}
 \Gamma'_{\pm}(\bsx) = \prod_{i=1}^\infty\Gamma'_{\pm}(x_i),\qquad
 \Gamma'_{\pm}(z) = \exp\left(- \sum_{k=1}^\infty\frac{(-z)^k}{k}J_{\pm k}\right)
\end{gather*}
of another type \cite{BY08} as
\begin{gather*}
 \Gamma^\sigma_{\pm}(\bsx) = \begin{cases}
 \Gamma_{\pm}(\bsx) &\text{if $\sigma = +1$},\\
 \Gamma'_{\pm}(\bsx) &\text{if $\sigma = -1$}.
 \end{cases}
\end{gather*}
The amplitude for $N = 1$ (with $\sigma_1 = 1$) agrees with the single topological vertex $C_{\lambda\mu\nu}$:
\begin{gather*}
 Z^{\lambda\mu}_\nu = C_{\lambda\mu\nu} = q^{\kappa(\mu)/2}s_{\tp{\nu}}\big(q^{-\rho}\big)
 \sum_{\eta\in\calP}s_{\tp{\lambda}/\eta}\big(q^{-\nu-\rho}\big) s_{\mu/\eta}\big(q^{-\tp{\nu}-\rho}\big).
\end{gather*}

The formula (\ref{Z-strip-IKP}) of Iqbal and Kashani-Poor can be recovered from~(\ref{Z-strip-gen}) as well. If $\alpha_0 = \alpha_N = \varnothing$, (\ref{Z-strip-gen})~simplif\/ies as
\begin{gather*}
 Z^{\varnothing\varnothing}_{\beta_1\cdots\beta_N} = s_{\tp{\beta}_1}\big(q^{-\rho}\big)\cdots s_{\tp{\beta}_N}\big(q^{-\rho}\big)
\big\langle 0\,|\, \Gamma^{\sigma_1}_{-}\big(q^{-\beta^{(1)}-\rho}\big)
 \Gamma^{\sigma_1}_{+}\big(q^{-\tp{\beta}^{(1)}-\rho}\big)(\sigma_1Q_1\sigma_2)^{L_0} \cdots \notag\\
\hphantom{Z^{\varnothing\varnothing}_{\beta_1\cdots\beta_N} =}{}\times
 \Gamma^{\sigma_{N-1}}_{-}\big(q^{-\beta^{(N-1)}-\rho}\big)
 \Gamma^{\sigma_{N-1}}_{+}\big(q^{-\tp{\beta}^{(N-1)}-\rho}\big)
 (\sigma_{N-1}Q_{N-1}\sigma_N)^{L_0} \notag\\
\hphantom{Z^{\varnothing\varnothing}_{\beta_1\cdots\beta_N} =}{}
\times
 \Gamma^{\sigma_N}_{-}\big(q^{-\beta^{(N)}-\rho}\big)
 \Gamma^{\sigma_N}_{+}\big(q^{-\tp{\beta}^{(N)}-\rho}\big)\,|\,0 \big\rangle.
\end{gather*}
One can use the commutation relations \cite{ORV03, BY08}
\begin{gather} \Gamma_{+}(\bsx)\Gamma_{-}(\bsy)= \prod_{i,j=1}^\infty(1 - x_iy_j)^{-1} \cdot\Gamma_{-}(\bsy)\Gamma_{+}(\bsx),\notag\\
 \Gamma'_{+}(\bsx)\Gamma'_{-}(\bsy)= \prod_{i,j=1}^\infty(1 - x_iy_j)^{-1} \cdot\Gamma'_{-}(\bsy)\Gamma'_{+}(\bsx),\notag\\
 \Gamma_{+}(\bsx)\Gamma'_{-}(\bsy)= \prod_{i,j=1}^\infty(1 + x_iy_j) \cdot\Gamma'_{-}(\bsy)\Gamma_{+}(\bsx),\notag\\
 \Gamma'_{+}(\bsx)\Gamma_{-}(\bsy)= \prod_{i,j=1}^\infty(1 + x_iy_j) \cdot\Gamma_{-}(\bsy)\Gamma'_{+}(\bsx) \label{Gamma-com-rel}
\end{gather}
to move $\Gamma^\sigma_{+}$'s to the right until they hit $|0\rangle$ and disappears. What remains are~$\Gamma^\sigma_{-}$'s, which hit~$\langle 0|$ and disappear. Contributions of $c$-number factors in~(\ref{Gamma-com-rel}) give rise to the inf\/inite products in~(\ref{Z-strip-IKP}).

\subsection{Generating functions for vertical legs}

Let us consider the generating function
\begin{gather*}
 \tau_n(\bsx) = \sum_{\beta_n\in\calP}s_{\tp{\beta}_n}(\bsx)Z_{n,\beta_n}/Z,
\end{gather*}
where
\begin{gather*}
 Z_{n,\beta_n} = Z^{\varnothing\varnothing}_{\cdots\varnothing\beta_n\varnothing\cdots},\qquad
 Z = Z^{\varnothing\varnothing}_{\varnothing\cdots\varnothing},
\end{gather*}
of the normalized open string amplitudes with respect to the partition $\beta_n$ on the $n$-th vertical leg. As we show below, $\tau_n(\bsx)$ has a fermionic expression that implies that $\tau_n(\bsx)$ is a special KP tau function of the type shown in~(\ref{tauh-gen}).

The unnormalized amplitudes are a specialization of (\ref{Z-strip-IKP}):
\begin{gather*}
 Z_{n,\beta_n} = s_{\tp{\beta}_n}\big(q^{-\rho}\big) \prod_{m<n}\prod_{i,j=1}^\infty
 \big(1 - Q_{mn}q^{-\beta^{(n)}_i+i+j-1}\big)^{-\sigma_m\sigma_n}\notag\\
\hphantom{Z_{n,\beta_n} =}{}\times \prod_{m>n}\prod_{i,j=1}^\infty \big(1 - Q_{nm}q^{-\tp{\beta}^{(n)}_i+i+j-1}\big)^{-\sigma_n\sigma_m}.
\end{gather*}
The normalized amplitudes $Z_{n,\beta_n}/Z$ can be cast into a fermionic expression with the aid of~(\ref{V(k)_0}) and its variant~\cite{TN15}
\begin{gather}
 \prod_{i,j=1}^\infty\dfrac{1 - Qq^{-\tp{\lambda}_i+i+j-1}}{1 - Qq^{i+j-1}}
 = \left\langle\tp{\lambda}\,|\, \exp\left(- \sum_{k=1}^\infty\frac{Q^kq^k}{k(1-q^k)}V^{(-k)}_0\right)\,|\,\tp{\lambda}\right\rangle \label{V(-k)_0}
\end{gather}
as follows.

\begin{Lemma}
The normalized amplitudes can be expressed as
\begin{gather}
 Z_{n,\beta_n}/Z = \big\langle\tp{\beta}_n\,|\,h\Gamma_{-}(q^{-\rho})\,|\,0\big\rangle, \label{Z(n)/norm}
\end{gather}
where
\begin{gather}
 h = \exp\left(- \sum_{m<n}\sigma_m\sum_{k=1}^\infty \frac{Q_{mn}^k}{k(1-q^k)}V^{(k)}_0
 + \sum_{m>n}\sigma_m\sum_{k=1}^\infty \frac{Q_{nm}^kq^k}{k(1-q^k)}V^{(-k)}_0\right) \label{h-sigma=+1}
\end{gather}
in the case where $\sigma_n = +1$, and
\begin{gather}
 h = \exp\left(- \sum_{m<n}\sigma_m\sum_{k=1}^\infty \frac{Q_{mn}^kq^k}{k(1-q^k)}V^{(-k)}_0
 + \sum_{m>n}\sigma_m\sum_{k=1}^\infty \frac{Q_{nm}^k}{k(1-q^k)}V^{(k)}_0\right) \label{h-sigma=-1}
\end{gather}
in the case where $\sigma_n = -1$.
\end{Lemma}

\begin{proof}
First consider the case where $\sigma_n = +1$. The normalized amplitudes in this case can be expressed as
\begin{gather*}
 Z_{n,\beta_n}/Z = s_{\tp{\beta}_n}(q^{-\rho})\prod_{m<n}
 \prod_{i,j=1}^\infty\frac{(1 - Q_{mn}q^{-\beta_{n,i}+i+j-1})^{-\sigma_m}}
 {(1 - Q_{mn}q^{i+j-1})^{-\sigma_m}}\\
\hphantom{Z_{n,\beta_n}/Z =}{} \times \prod_{m>n} \prod_{i,j=1}^\infty\frac{(1 - Q_{nm}q^{-\tp{\beta}_{n,i}+i+j-1})^{-\sigma_m}}
 {(1 - Q_{nm}q^{i+j-1})^{-\sigma_m}}.
\end{gather*}
One can apply (\ref{V(k)_0}) to the quotients in the product over $m < n$ and~(\ref{V(-k)_0}) to the quotient in the product over $m > n$. This yields an expression of the form~(\ref{Z(n)/norm}) with $h$ def\/ined as~(\ref{h-sigma=+1}). Next examine the case where $\sigma_n = -1$. The normalized amplitudes changes as
\begin{gather*}
 Z_{n,\beta_n}/Z = s_{\tp{\beta}_n}(q^{-\rho})\prod_{m<n} \prod_{i,j=1}^\infty\frac{(1 - Q_{mn}q^{-\tp{\beta}_{n,i}+i+j-1})^{\sigma_m}}
 {(1 - Q_{mn}q^{i+j-1})^{\sigma_m}}\\
\hphantom{Z_{n,\beta_n}/Z =}{} \times \prod_{m>n} \prod_{i,j=1}^\infty\frac{(1 - Q_{nm}q^{-\beta_{n,i}+i+j-1})^{\sigma_m}} {(1 - Q_{nm}q^{i+j-1})^{\sigma_m}}.
\end{gather*}
One can derive an expression of the form (\ref{Z(n)/norm}) in the same way with the operator $h$ def\/ined as~(\ref{h-sigma=-1}).
\end{proof}

The fermionic expression (\ref{Z(n)/norm}) of the normalized amplitudes imply the following conclusion.

\begin{Theorem}
$\tau_n(\bsx)$ is a tau function of the KP hierarchy with the fermionic expression
\begin{gather}
 \tau_n(\bsx) = \big\langle 0\,|\,\Gamma_{+}(\bsx)h\Gamma_{-}(q^{-\rho})\,|\,0\rangle, \label{taun-strip}
\end{gather}
where $h$ is the operator shown in \eqref{h-sigma=+1} and \eqref{h-sigma=-1}.
\end{Theorem}

\begin{Remark}
$h$ is an operator of the form (\ref{h-gen}). The parameters $h_j$, $r_j$ can be computed as follows:
\begin{itemize}\itemsep=0pt
\item[(a)] In the case where $\sigma_n = +1$,
\begin{gather}
 h_j = \exp\left(- \sum_{m<n}\sigma_m\sum_{k=1}^\infty \frac{Q_{mn}^k}{k(1-q^k)}q^{kj} + \sum_{m>n}\sigma_m\sum_{k=1}^\infty
 \frac{Q_{nm}^kq^k}{k(1-q^k)}q^{-kj}\right) \notag\\
\hphantom{h_j}{} = \prod_{m<n}\exp\left(- \sigma_m\sum_{i,k=1}\frac{(Q_{mn}q^{i-1+j})^k}{k}
 \right)\cdot \prod_{m>n}\exp\left(\sigma_m\sum_{i,k=1}^\infty\frac{(Q_{nm}q^{i-j})^k}{k} \right) \notag\\
\hphantom{h_j}{}= \prod_{m<n}\prod_{i=1}^\infty\big(1 - Q_{mn}q^{i-1+j}\big)^{\sigma_m}
 \cdot\prod_{m>n}\prod_{i=1}^\infty\big(1 - Q_{nm}q^{i-j}\big)^{-\sigma_m} \label{hj-sigma=+1}
\end{gather}
and
\begin{gather*}
 r_j = \prod_{m<n}\big(1 - Q_{mn}q^{j-1}\big)^{-\sigma_m} \cdot\prod_{m>n}\big(1 - Q_{nm}q^{1-j}\big)^{-\sigma_m}.
\end{gather*}
\item[(b)] In the case where $\sigma_n = -1$,
\begin{gather}
 h_j = \exp\left(- \sum_{m<n}\sigma_m\sum_{k=1}^\infty \frac{Q_{mn}^kq^k}{k(1-q^k)}q^{-kj} + \sum_{m>n}\sigma_m\sum_{k=1}^\infty
 \frac{Q_{nm}^k}{k(1-q^k)}q^{kj}\right) \notag\\
\hphantom{h_j}{} = \prod_{m<n}\exp\left(- \sigma_m\sum_{i,k=1}\frac{(Q_{mn}q^{i-j})^k}{k}
 \right)\cdot
 \prod_{m>n}\exp\left(\sigma_m\sum_{i,k=1}^\infty\frac{(Q_{nm}q^{i-1+j})^k}{k} \right) \notag\\
\hphantom{h_j}{}= \prod_{m<n}\prod_{i=1}^\infty\big(1 - Q_{mn}q^{i-j}\big)^{\sigma_m}
 \cdot\prod_{m>n}\prod_{i=1}^\infty\big(1 - Q_{nm}q^{i-1+j}\big)^{-\sigma_m} \label{hj-sigma=-1}
\end{gather}
and
\begin{gather*}
 r_j = \prod_{m<n}\big(1 - Q_{mn}q^{1-j}\big)^{\sigma_m} \cdot\prod_{m>n}\big(1 - Q_{nm}q^{j-1}\big)^{\sigma_m}.
\end{gather*}
\end{itemize}
This result shows that the tau functions $\tau_n(\bsx)$ are examples of Orlov and Scherbin's $q$-hypergeo\-met\-ric tau functions \cite{OS01a,OS01b}.
\end{Remark}

\subsection{Admissible basis and Kac--Schwarz operators}

(\ref{taun-strip}) is a somewhat complicated, but rather straightforward generalization of the tau function~(\ref{tau-conifold}) for the resolved conifold. One can obtain the generating operator $G$ of an admissible basis of~$W$ and the associated $q$-dif\/ference Kac--Schwarz operators~$A$,~$B$ in much the same way as the case of the resolved conifold.

One can derive the generating operator $G$ from $g = h\Gamma_{-}(q^{-\rho})$ by replacing
\begin{gather*}
 V^{(k)}_0 \to q^{kD}, \qquad \Gamma_{-}(q^{-\rho}) \to \prod_{i=1}^\infty\big(1 - q^{i-1/2}x\big)^{-1}.
\end{gather*}
Thus the generating operator $G$, like (\ref{G-conifold}), takes the factorized form
\begin{gather}
 G = H\cdot\prod_{i=1}^\infty\big(1 - q^{i-1/2}x\big)^{-1}, \label{G-strip}
\end{gather}
where $H$ is the following avatar of $h$:
\begin{itemize}\itemsep=0pt
\item[(a)]In the case where $\sigma_n = +1$,
\begin{gather}
 H = \prod_{m<n}\prod_{i=1}^\infty\big(1 - Q_{mn}q^{i-1+D}\big)^{\sigma_m}\cdot
 \prod_{m>n}\prod_{i=1}^\infty\big(1 - Q_{nm}q^{i-D}\big)^{-\sigma_m}. \label{H-sigma=+1}
\end{gather}
\item[(b)]In the case where $\sigma_n = -1$,
\begin{gather}
 H = \prod_{m<n}\prod_{i=1}^\infty\big(1 - Q_{mn}q^{i-D}\big)^{\sigma_m}\cdot
 \prod_{m>n}\prod_{i=1}^\infty\big(1 - Q_{nm}q^{i-1+D}\big)^{-\sigma_m}. \label{H-sigma=-1}
\end{gather}
\end{itemize}
Note that (\ref{H-sigma=+1}) and (\ref{H-sigma=-1}) have the same structure as the expressions (\ref{hj-sigma=+1}) and (\ref{hj-sigma=-1}) of the parameters $h_j$ of the operator~$h$.

Let $\{\Phi_j(x)\}_{j=0}^\infty$ denote the admissible basis generated by~$G$:
\begin{gather*}
 \Phi_j(x)
 = H\left(x^{-j}\prod_{i=1}^\infty\big(1 - q^{i-1/2}x\big)^{-1}\right) = \sum_{k=1}^\infty\frac{q^{k/2}}{(q;q)_k}Hx^{k-j}.
\end{gather*}
Rewriting this expression further, we obtain the following result.

\begin{Theorem}
The basis elements $\Phi_j(x)$ are functions of the following form:
\begin{itemize}\itemsep=0pt
\item[{\rm (a)}]In the case where $\sigma_n = +1$,
\begin{gather*}
 \Phi_j(x) = \prod_{m<n}\prod_{i=1}^\infty\big(1 - Q_{mn}q^{i-1-j}\big)^{\sigma_m}\cdot \prod_{m>n}\prod_{i=1}^\infty\big(1 - Q_{nm}q^{i+j}\big)^{-\sigma_m}\notag\\
\hphantom{\Phi_j(x)}{}\times \sum_{k=0}^\infty
 \frac{\prod\limits_{m<n}(Q_{mn}q^{-j};q)_k^{-\sigma_m} \prod\limits_{m>n}(Q_{nm}q^j;q^{-1})_k^{-\sigma_m}q^{k/2}}{(q;q)_k}x^{k-j}.
\end{gather*}
\item[{\rm (b)}]In the case where $\sigma_n = -1$,
\begin{gather*}
 \Phi_j(x) = \prod_{m<n}\prod_{i=1}^\infty\big(1 - Q_{mn}q^{i+j}\big)^{\sigma_m}\cdot
 \prod_{m>n}\prod_{i=1}^\infty\big(1 - Q_{nm}q^{i-1-j}\big)^{-\sigma_m}\notag\\
\hphantom{\Phi_j(x)}{}\times \sum_{k=0}^\infty
 \frac{\prod\limits_{m<n}(Q_{mn}q^j;q^{-1})_k^{\sigma_m} \prod\limits_{m>n}(Q_{nm}q^{-j};q)_k^{\sigma_m}q^{k/2}}{(q;q)_k}x^{k-j}.
\end{gather*}
\end{itemize}
\end{Theorem}

As remarked in the case of the resolved conifold, see~(\ref{Phij-conifold}), the essential part of $\Phi_j(x)$'s are $q$-hypergeometric functions. In particular, $\Phi_0(x)$ agrees, up to a constant multiplier, with the generating function
\begin{gather}
 \Psi(x) = \sum_{k=0}^\infty x^k Z_{n,(1^k)}/Z
 \label{Psi-strip}
\end{gather}
of the normalized open string amplitudes considered in our previous work~\cite{Takasaki13}.

Let us proceed to computation of the $q$-dif\/ference Kac--Schwarz operators~$A$,~$B$. Since~$G$ has the partially factorized form~(\ref{G-strip}) and the second factor has the property
\begin{gather*}
 \prod_{i=1}^\infty\big(1 - q^{i-1/}x\big)^{-1}\cdot q^{-D} \cdot\prod_{i=1}^\infty\big(1 - q^{i-1/2}x\big) = q^{-D}\big(1 - q^{1/2}x\big),
\end{gather*}
see (\ref{A-CP3}), we have the following preliminary expressions of $A$ and $B$:
\begin{gather*}
 A = q^{-D}H\cdot \big(1 - q^{1/2}x\big)\cdot H^{-1},\qquad B = H\cdot x^{-1}\cdot H^{-1}.
\end{gather*}
Thus our main task is to compute $H\cdot x^{\pm 1}\cdot H^{-1}$.

\begin{Lemma}
\begin{gather*}
 H\cdot x\cdot H^{-1} = xR,
\end{gather*}
where
\begin{gather}
 R = \prod_{m<n}\big(1 - Q_{mn}q^{\sigma_n D}\big)^{-\sigma_m\sigma_n} \cdot\prod_{m>n}\big(1 - Q_{nm}q^{-\sigma_n D}\big)^{-\sigma_n\sigma_m}. \label{R-def}
\end{gather}
\end{Lemma}

\begin{proof}
By the identity (\ref{x-conj}), one can move $x$ to the left.
In the case where $\sigma_n = +1$, this computation proceeds
as follows:
\begin{gather*}
 H\cdot x\cdot H^{-1} = x \prod_{m<n}\prod_{i=1}^\infty\big(1 - Q_{mn}q^{i+D}\big)^{\sigma_m}
 \cdot\prod_{m>n}\prod_{i=1}^\infty\big(1 - Q_{nm}q^{i-1-D}\big)^{-\sigma_m}\\
\hphantom{H\cdot x\cdot H^{-1} =} {}\times
 \prod_{m>n}\prod_{i=1}^\infty\big(1 - Q_{nm}q^{i-D}\big)^{\sigma_m} \cdot\prod_{m<n}\prod_{i=1}^\infty\big(1 - Q_{mn}q^{i-1+D}\big)^{-\sigma_m}\\
\hphantom{H\cdot x\cdot H^{-1}} {}= x \prod_{m<n}\big(1 - Q_{mn}q^D\big)^{-\sigma_m} \cdot\prod_{m>n}\big(1 - Q_{nm}q^{-D}\big)^{-\sigma_m}.
\end{gather*}
The case where $\sigma_n = -1$ can be treated in the same way:
\begin{gather*}
 H\cdot x\cdot H^{-1} = x \prod_{m<n}\prod_{i=1}^\infty\big(1 - Q_{mn}q^{i-1-D}\big)^{\sigma_m}
 \cdot\prod_{m>n}\prod_{i=1}^\infty\big(1 - Q_{nm}q^{i+D}\big)^{-\sigma_m}\\
\hphantom{H\cdot x\cdot H^{-1} =}{}\times
 \prod_{m>n}\prod_{i=1}^\infty\big(1 - Q_{nm}q^{i-1+D}\big)^{\sigma_m} \cdot\prod_{m<n}\prod_{i=1}^\infty\big(1 - Q_{mn}q^{i-D}\big)^{-\sigma_m}\\
\hphantom{H\cdot x\cdot H^{-1}}{}= x \prod_{m<n}\big(1 - Q_{mn}q^{-D}\big)^{\sigma_m} \cdot\prod_{m>n}\big(1 - Q_{nm}q^D\big)^{\sigma_m}.\tag*{\qed}
\end{gather*}\renewcommand{\qed}{}
\end{proof}

By this lemma applied to the preliminary expression, we f\/ind the following explicit form of~$A$ and~$B$.

\begin{Theorem}The $q$-difference Kac--Schwarz operators can be expressed as
\begin{gather*}
 A = q^{-D}\big(1 - q^{1/2}xR\big),\qquad B = (xR)^{-1} = R^{-1}\cdot x^{-1},
\end{gather*}
where $R$ is the operator defined in~\eqref{R-def}.
\end{Theorem}

This is a generalization of the result for the resolved conifold, see (\ref{A-conifold}) and (\ref{B-conifold}), for which $R = 1 - Qq^D$. The $q$-dif\/ference equation (\ref{Phi0-eq-conifold}) for $\Phi_0(x)$ is generalized as
\begin{gather}
 \big(1 - q^{1/2}xR\big)\Phi_0(x) = \Phi_0(qx). \label{Phi0-eq-strip}
\end{gather}
$R$ can be expressed as
\begin{gather*}
 R = \frac{Q(q^D)}{P(q^D)},
\end{gather*}
where $P(q^D)$ and $Q(q^D)$ are Laurent polynomials $q^D$:
\begin{gather*}
 P(q^D) = \prod_{m<n,\sigma_m\sigma_n>0}\big(1 - Q_{mn}q^{\sigma_nD}\big) \cdot\prod_{m>n,\sigma_m\sigma_n>0}\big(1 - Q_{nm}q^{-\sigma_nD}\big),\\
 Q(q^D) = \prod_{m<n,\sigma_m\sigma_n<0}\big(1 - Q_{mn}q^{\sigma_nD}\big) \cdot\prod_{m>n,\sigma_m\sigma_n<0}\big(1 - Q_{nm}q^{-\sigma_nD}\big).
\end{gather*}
(\ref{Phi0-eq-strip}) agrees with the $q$-dif\/ference equation derived for the generating function (\ref{Psi-strip}) in our previous work \cite{Takasaki13}.

\subsection{Generating functions for non-vertical legs}

Let us now consider the generating functions
\begin{gather*}
 \tau_0(\bsx) = \sum_{\alpha_0\in\calP}s_{\tp{\alpha}_0}(\bsx)Z_{0,\alpha_0}/Z,\qquad Z_{0,\alpha_0} = Z^{\alpha_0\varnothing}_{\varnothing\cdots\varnothing},\\
 \tau_{N+1}(\bsx) = \sum_{\alpha_N\in\calP}s_{\tp{\alpha}_N}(\bsx)Z_{N+1,\alpha_N}/Z,\qquad
 Z_{N+1,\alpha_N} = Z^{\varnothing\alpha_N}_{\varnothing\cdots\varnothing},
\end{gather*}
of the normalized open string amplitudes for the leftmost and rightmost legs. As it turns out, their structure is quite distinct from the foregoing generating functions $\tau_1(\bsx),\ldots,\tau_N(\bsx)$.

The unnormalized amplitudes are a specialization of (\ref{Z-strip-gen}):
\begin{gather}
 Z_{0,\alpha_0} = q^{(1-\sigma_1)\kappa(\alpha_0)/4} \big\langle\tp{\alpha}_0\,|\,\Gamma^{\sigma_1}_{-}\big(q^{-\rho}\big)
 \Gamma^{\sigma_1}_{+}\big(q^{-\rho}\big)(\sigma_1Q_1\sigma_2)^{L_0} \cdots \notag\\
\hphantom{Z_{0,\alpha_0} =}{}\times (\sigma_{N-1}Q_{N-1}\sigma_N)^{L_0} \Gamma^{\sigma_N}_{-}\big(q^{-\rho}\big)
 \Gamma^{\sigma_N}_{+}\big(q^{-\rho}\big)\,|\,0\big\rangle, \label{Z(0)}
\\
 Z_{N+1,\alpha_N} = q^{(1+\sigma_N)\kappa(\alpha_N)/4} \big\langle 0\,|\,\Gamma^{\sigma_1}_{-}\big(q^{-\rho}\big)
 \Gamma^{\sigma_1}_{+}\big(q^{-\rho}\big)(\sigma_1Q_1\sigma_2)^{L_0} \cdots \notag\\
\hphantom{Z_{N+1,\alpha_N} =}{}\times (\sigma_{N-1}Q_{N-1}\sigma_N)^{L_0} \Gamma^{\sigma_N}_{-}\big(q^{-\rho}\big)
 \Gamma^{\sigma_N}_{+}(q^{-\rho})|\alpha_N\rangle. \label{Z(N+1)}
\end{gather}
These amplitudes can be converted to a reduced form by the following procedure:
\begin{itemize}\itemsep=0pt
\item[1.] $\kappa(\lambda)$ is the matrix elements of the diagonal operator
\begin{gather*}
 K = \sum_{n\in\ZZ}(n-1/2)^2{:}\psi_{-n}\psi^*_n{:}
\end{gather*}
on the fermionic Fock space in the sense that
\begin{gather*}
 \kappa(\lambda) = \langle\lambda|K|\lambda\rangle.
\end{gather*}
Consequently, the $c$-number prefactors $q^{(1-\sigma_1)\kappa(\alpha_0)/4}$ and $q^{(1+\sigma_N)\kappa(\alpha_N)/4}$ can be moved inside $\langle\tp{\alpha}_0|$ and $|\alpha_N\rangle$ as the operators $q^{-(1-\sigma_1)K/4}$ and $q^{(1+\sigma_N)K/4}$. Note that the sign of the exponent
of the f\/irst factor is f\/lipped because of the identity $\kappa(\tp{\alpha}_0) = - \kappa(\alpha_0)$.
\item[2.]
Operators of the form $Q^{L_0}$ in (\ref{Z(0)}) can be moved to the right until they hit the vacuum vector and disappear. Those in (\ref{Z(N+1)}) can be eliminated in the same way. In the course of this procedure, the vertex operators are modif\/ied as
\begin{gather*}
 Q^{L_0}\Gamma^\sigma_{-}(\bsx) = \Gamma^\sigma_{-}(Q\bsx)Q^{L_0},\qquad \Gamma^\sigma_{+}(\bsx)Q^{L_0} = Q^{L_0}\Gamma^\sigma_{+}(Q\bsx).
\end{gather*}
\item[3.]
$\Gamma^\sigma_{-}$'s in (\ref{Z(0)}) and $\Gamma^\sigma_{+}$'s in (\ref{Z(N+1)}) can be moved until they hit the vacuum vectors and disappear. Exchanging the order of vertex operators yields $c$-number factors as the commutation relations~(\ref{Gamma-com-rel}) show. These $c$-number factors do not depend on $\alpha_,\alpha_N$ and cancel out in the normalized amplitudes.
\end{itemize}
Thus the following reduced expressions of the normalized amplitudes are obtained:
\begin{gather*}
 Z_{0,\alpha_0}/Z = \left\langle\tp{\alpha}_0\,|\,q^{-(1-\sigma_1)K/4}\Gamma^{\sigma_1}_{-}\big(q^{-\rho}\big)
 \prod_{n=2}^N\Gamma^{\sigma_n}_{-}\big(\sigma_1Q_{1n}\sigma_nq^{-\rho}\big) \,|\,0\right\rangle,\\
 Z_{N+1,\alpha_N}/Z= \left\langle 0\,|\prod_{n=1}^{N-1}\Gamma^{\sigma_n}_{+}\big(\sigma_nQ_{nN}\sigma_Nq^{-\rho}\big)
 \cdot\Gamma^{\sigma_N}_{+}\big(q^{-\rho}\big)q^{(1+\sigma_N)K/4}\,|\,\alpha_N\right\rangle.
\end{gather*}
This implies that the generating functions $\tau_0(\bsx)$, $\tau_{N+1}(\bsx)$ of these normalized amplitudes are tau functions of the KP hierarchy:

\begin{Theorem}$\tau_0(\bsx)$ and $\tau_{N+1}(\bsx)$ are tau functions of the KP hierarchy with the following fermionic expression:
\begin{gather}
 \tau_0(\bsx) = \left\langle 0\,|\,\Gamma_{+}(\bsx)q^{-(1-\sigma_1)K/4}\Gamma^{\sigma_1}_{-}\big(q^{-\rho}\big)
 \prod_{n=2}^N\Gamma^{\sigma_n}_{-}\big(\sigma_1Q_{1n}\sigma_nq^{-\rho}\big)\,|\,0\right\rangle, \label{tau0-strip}\\
 \tau_{N+1}(\bsx) = \left\langle 0\,|\prod_{n=1}^{N-1}\Gamma^{\sigma_n}_{+}\big(\sigma_nQ_{nN}\sigma_Nq^{-\rho}\big)
 \cdot\Gamma^{\sigma_N}_{+}\big(q^{-\rho}\big)q^{(1+\sigma_N)K/4} \Gamma'_{-}(\bsx)\,|\,0\right\rangle. \label{tauN+1-strip}
\end{gather}
\end{Theorem}

\begin{Remark}
One can rewrite the fermionic expression of $\tau_{N+1}(\bsx)$ as
\begin{gather}
 \tau_{N+1}(\bsx)  = \left\langle 0\,|\,\Gamma_{+}(\bsx)q^{-(1+\sigma_N)K/4} \Gamma^{-\sigma_N}_{-}\big(q^{-\rho}\big)\prod_{n=1}^{N-1}
 \Gamma^{-\sigma_n}_{-}\big(\sigma_NQ_{Nn}\sigma_{n}q^{-\rho}\big) \,|\,0\right\rangle, \label{tauN+1-strip2}
\end{gather}
where $Q_{nm}$ for $n > m$ denotes $Q_nQ_{n-1}\cdots Q_{m-1}$. This shows that $\tau_{N+1}(\bsx)$ may be thought of as the generating function for the leftmost leg of the web diagram rotated 180 degrees.
\end{Remark}

Although the structure of (\ref{tau0-strip}), (\ref{tauN+1-strip}) and (\ref{tauN+1-strip2}) are distinct from the tau functions~(\ref{taun-strip}) for the vertical legs, $q$-dif\/ference Kac--Schwarz operators can be derived in much the same way. We here present the result for $\tau_0(\bsx)$ only. $\tau_{N+1}(\bsx)$ can be treated in the same way.

The generating operator $G$ for $\tau_0(\bsx)$ can be obtained as an avatar of the operator
\begin{gather*}
 g = q^{-(1-\sigma_1)K/4}\Gamma^{\sigma_1}_{-}\big(q^{-\rho}\big) \prod_{n=2}^N\Gamma^{\sigma_n}_{-}\big(\sigma_1Q_{1n}\sigma_nq^{-\rho}\big)
\end{gather*}
on the fermionic Fock space. In view of the correspondence
\begin{gather*}
 K \longleftrightarrow (\Delta - 1/2)^2 \longleftrightarrow (D - 1/2)^2,
\end{gather*}
the f\/irst factor $q^{-(1-\sigma_1)K/4}$ corresponds to the operator $q^{-(1-\sigma_1)(D-1/2)^2/4}$ on $V$. On the other hand, just as $\Gamma_{-}(q^{-\rho})$ is translated to~(\ref{G-CP3}), vertex operators of the form $\Gamma^\sigma_{-}(Qq^{-\rho})$ corresponds to the multiplication operators
\begin{gather*}
 \exp\left(\sigma\sum_{k=1}^\infty  \frac{(\sigma Qq^{1/2})^k}{k(1-q^k)}x^k\right) = \prod_{i=1}^\infty\big(1 - \sigma Qq^{i-1/2}\big)^{-\sigma}.
\end{gather*}
Thus $G$ turns out to take the following form:
\begin{gather*}
 G = q^{-(1-\sigma_1)(D-1/2)^2/4} \prod_{i=1}^\infty \big(1 - \sigma_1q^{i-1/2}x\big)^{-\sigma_1}\cdot \prod_{n=2}^N\prod_{i=1}^\infty
 \big(1 - \sigma_1Q_{1n}q^{i-1/2}x\big)^{-\sigma_n}.
\end{gather*}

The associated admissible basis of $W$ consists of the following functions:
\begin{itemize}\itemsep=0pt
\item[(a)]In the case where $\sigma_1 = +1$,
\begin{gather*}
 \Phi_j(x) = x^{-j}\prod_{i=1}^\infty \big(1 - q^{i-1/2}x\big)^{-1}\cdot \prod_{n=2}^N\prod_{i=1}^\infty \big(1 - Q_{1n}q^{i-1/2}x\big)^{-\sigma_n},
\end{gather*}
\item[(b)]in the case where $\sigma_1 = - 1$,
\begin{gather*}
 \Phi_j(x) = q^{-(D-1/2)^2/2} \left(x^{-j}\prod_{i=1}^\infty \big(1 + q^{i-1/2}x\big)\cdot
 \prod_{n=2}^N\prod_{i=1}^\infty \big(1 + Q_{1n}q^{i-1/2}x\big)^{-\sigma_n}\right).
\end{gather*}
\end{itemize}

These results show that the structure of the generating operator and the admissible basis, too, is entirely dif\/ferent from the case of vertical legs shown in~(\ref{G-strip}). In particular, the essential part of $\Phi_j(x)$'s are composite quantum dilogarithmic functions. In this sense, the situation is rather similar to the case of~$\CC^3$, see~(\ref{Phij-CP3}).

Computation of the $q$-dif\/ference Kac--Schwarz operators is more or less parallel to the case of~$\CC^3$ except that the avatar $q^{(D-1/2)^2/2}$ of $q^{K/2}$ joins the game.

\begin{Theorem}The $q$-difference Kac--Schwarz operators $A$, $B$ for $\tau_0(\bsx)$ can be expressed as
\begin{gather*}
 A = q^{-D}\big(1 - q^{1/2}x\big)\prod_{n=2}^N\big(1 - Q_{1n}q^{1/2}x\big)^{\sigma_n},\qquad B = x^{-1}
\end{gather*}
in the case where $\sigma_1 = +1$, and
\begin{gather*}
 A = q^{-D}\big(1 + q^{1/2}xq^{-D}\big)^{-1} \prod_{n=2}^N\big(1 + Q_{1n}q^{1/2}xq^{-D}\big)^{\sigma_n},\qquad B = \big(xq^{-D}\big)^{-1} = q^Dx^{-1}
\end{gather*}
in the case where $\sigma_1 = -1$.
\end{Theorem}

\begin{proof}
Since the case where $\sigma_1 = +1$ is simpler, let us explain the case where $\sigma_1 = -1$. In this case, $G$ is a product of two operators as
\begin{gather*}
 G = q^{-(D-1/2)^2/2}G',
\end{gather*}
where
\begin{gather*}
 G' = \prod_{i=1}^\infty\big(1 + q^{i-1/2}x\big) \cdot\prod_{n=2}^N\prod_{i=1}^\infty\big(1 + Q_{1n}q^{i-1/2}x\big)^{-\sigma_n}.
\end{gather*}
As an intermediate step, one can compute $A' = G'\cdot q^{-D}\cdot G'^{-1}$ with the aid of~(\ref{q^D-conj}) as
\begin{gather*}
 A' = q^{-D}\prod_{i=1}^\infty\big(1 + q^{i+1/2}x\big) \cdot\prod_{n=2}^N\prod_{i=1}^\infty\big(1 + Q_{1n}q^{i+1/2}x\big)^{-\sigma_n}\\
\hphantom{A' =}{}\times
 \prod_{n=2}^N\prod_{i=1}^\infty\big(1 + Q_{1n}q^{i-1/2}x\big)^{\sigma_n} \cdot\prod_{i=1}^\infty\big(1 + q^{i-1/2}x\big)^{-1}\\
\hphantom{A'}{} = q^{-D}\big(1 + q^{1/2}x\big)^{-1}\prod_{n=2}^N\big(1 + Q_{1n}q^{1/2}x\big)^{\sigma_n}.
\end{gather*}
Thus $A = G\cdot q^{-D}\cdot G^{-1}$ can be expressed as
\begin{gather*}
 A = q^{-(D-1/2)^2/2}\cdot \big(1 + q^{1/2}x\big)^{-1}\prod_{n=2}^N\big(1 + Q_{1n}q^{1/2}x\big)^{\sigma_n}\cdot q^{(D-1/2)^2/2}.
\end{gather*}
This turns into the f\/inal form by the operator identity
\begin{gather*}
 q^{-(D-1/2)^2/2}f(x)q^{(D-1/2)^2/2} = f\big(xq^{-D}\big).\tag*{\qed}
\end{gather*}\renewcommand{\qed}{}
\end{proof}

\section{Closed topological vertex}\label{section4}

Let $Z_{\beta_1\beta_2}$ denote the open string amplitude of closed topological vertex with two stacks of branes as shown in Fig.~\ref{fig3}. $Q_1$, $Q_2$, $Q_3$ are the K\"ahler parameters assigned to the internal lines. The amplitude $Z_{\varnothing\varnothing}$ in the closed sector was computed by Bryan and Karp~\cite{BK03} and Su{\l}kowski~\cite{Sulkowski06}. In our previous work~\cite{TN15}, we computed $Z_{\beta_1,\beta_2}$ and obtained the following fermionic expression:
\begin{gather}
 Z_{\beta_1\beta_2} = q^{\kappa(\beta_2)/2}\prod_{i,j=1}^\infty \big(1 - Q_1Q_2q^{-\beta_{1,i}-\tp{\beta}_{2,j}+i+j-1}\big)^{-1}\notag\\
\hphantom{Z_{\beta_1\beta_2} =}{}\times
 \big\langle\tp{\beta}_1\,|\,\Gamma_{-}\big(q^{-\rho}\big)\Gamma_{+}\big(q^{-\rho}\big)(-Q_1)^{L_0}
 \Gamma'_{-}\big(q^{-\rho}\big)\Gamma'_{+}\big(q^{-\rho}\big)(-Q_3)^{L_0}\notag\\
\hphantom{Z_{\beta_1\beta_2} =}{}
\times \Gamma_{-}\big(q^{-\rho}\big)\Gamma_{+}\big(q^{-\rho}\big)(-Q_2)^{L_0}
 \Gamma'_{-}\big(q^{-\rho})\Gamma'_{+}\big(q^{-\rho}\big)\,|\,\tp{\beta}_2\big\rangle. \label{Z-ctv}
\end{gather}

\begin{figure}[t]\centering
\includegraphics[scale=0.7]{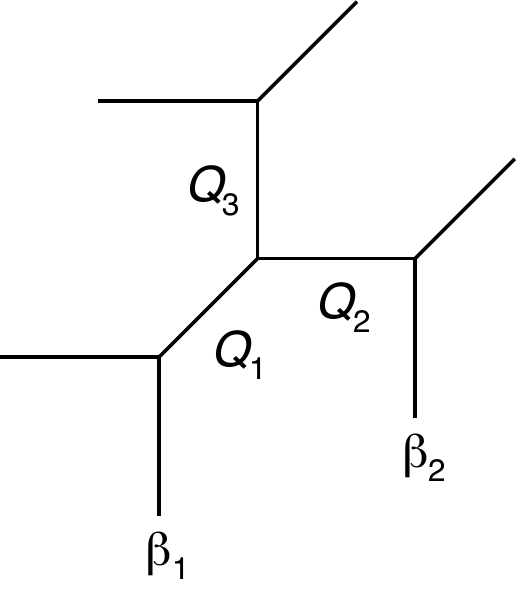}
\caption{Closed topological vertex with two stacks of branes.}\label{fig3}
\end{figure}

The main part $\langle\tp{\beta}_1|\cdots|\tp{\beta}_2\rangle$ of~(\ref{Z-ctv}) happens to be essentially the same as the open string amplitude of strip geometry shown in Fig.~\ref{fig4}. This part is corrected by the f\/irst two factors. In the previous work, we made use of this particular structure to derive $q$-dif\/ference equations for the generating functions{\samepage
\begin{gather}
 \Psi_1(x) = \sum_{k=0}^\infty x^kZ_{(1^k)\varnothing}/Z,\qquad \Psi_2(x) = \sum_{k=0}^\infty x^kZ_{\varnothing(1^k)}/Z,\qquad Z = Z_{\varnothing\varnothing}, \label{Psi-ctv}
\end{gather}
of normalized open string amplitudes.}

\begin{figure}[t]\centering
\includegraphics[scale=0.7]{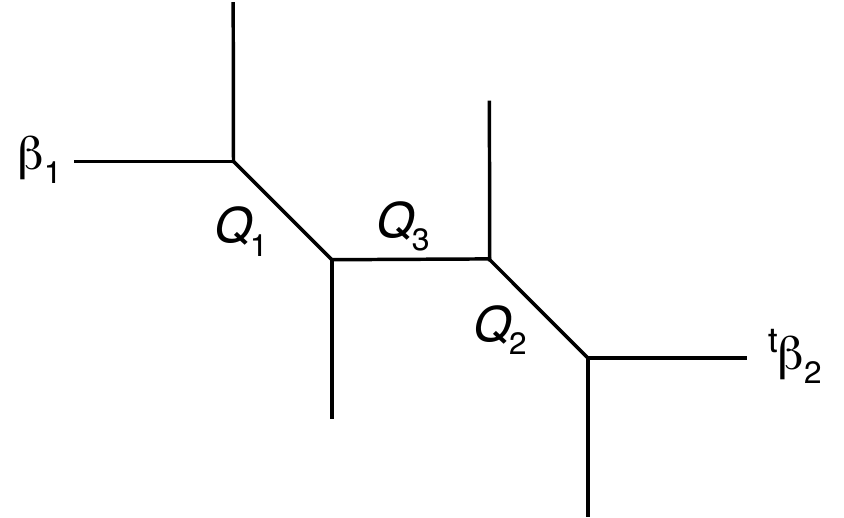}
\caption{New strip geometry emerging from $Z_{\beta_1\beta_2}$.}\label{fig4}
\end{figure}

\subsection{Generating functions as tau functions}

We can construct two multi-variate generating functions:
\begin{gather*}
 \tau_1(\bsx) = \sum_{\beta\in\calP}s_{\tp{\beta}}(\bsx)Z_{1,\beta}/Z,\qquad Z_{1,\beta} = Z_{\beta\varnothing},\\
 \tau_2(\bsx) = \sum_{\beta\in\calP}s_{\tp{\beta}}(\bsx)Z_{2,\beta}/Z,\qquad  Z_{2,\beta} = Z_{\varnothing\beta}.
\end{gather*}
As we show below, both $\tau_1(\bsx)$ and $\tau_2(\bsx)$ are tau functions of the KP hierarchy, and the asso\-ciated points of the Sato Grassmannian have admissible bases that are generated by $q$-dif\/ference operators. The reasoning is similar to the case the of generating functions of $Z_{0,\alpha_0}/Z$ and $Z_{N+1,\alpha_N}/Z$ in the last section.

Let us examine $Z_{1,\beta}/Z$ in detail. As explained in the case of strip geometry, one can eliminate the operators of the form $(-Q)^{L_0}$ and $\Gamma^\sigma_{-}$ from the unnormalized amplitude
\begin{gather*}
 Z_{1,\beta} = \prod_{i,j=1}^\infty\big(1 - Q_1Q_2q^{-\beta_i+i+j-1}\big)^{-1} \big\langle\tp{\beta}\,|\,\Gamma_{-}\big(q^{-\rho}\big)\Gamma_{+}\big(q^{-\rho}\big)(-Q_1)^{L_0}
 \Gamma'_{-}\big(q^{-\rho}\big)\\
\hphantom{Z_{1,\beta} =}{}\times \Gamma'_{+}\big(q^{-\rho}\big)(-Q_3)^{L_0}\Gamma_{-}\big(q^{-\rho}\big)\Gamma_{+}\big(q^{-\rho}\big)(-Q_2)^{L_0}
 \Gamma'_{-}\big(q^{-\rho}\big)\Gamma'_{+}\big(q^{-\rho}\big)\,|\,0\big\rangle
\end{gather*}
as
\begin{gather*}
 Z_{1,\beta} = (\text{constant independent of $\beta$}) \cdot \prod_{i,j=1}^\infty\big(1 - Q_1Q_2q^{-\beta_i+i+j-1}\big)^{-1}\\
\hphantom{Z_{1,\beta} =}{}\times \big\langle\tp{\beta}\,|\,\Gamma_{-}\big(q^{-\rho}\big)\Gamma'_{-}\big({-}Q_1q^{-\rho}\big)
 \Gamma_{-}\big(Q_1Q_3q^{-\rho}\big)\Gamma'_{-}\big({-}Q_1Q_2Q_3q^{-\rho}\big)\, |\,0\big\rangle,
\end{gather*}
The f\/irst constant factor disappears in the normalized amplitude. The contribution of the inf\/inite product can be computed by the formula (\ref{V(k)_0}). Thus the normalized amplitude can be expressed as
\begin{gather*}
 Z_{1,\beta}/Z = \big\langle\tp{\beta}\,|\,g_1\,|\,0\big\rangle,
\end{gather*}
where
\begin{gather}
 g_1 = \exp\left(- \sum_{k=1}^\infty\frac{(Q_1Q_2)^k}{k(1-q^k)}V^{(k)}_0 \right) \nonumber\\
 \hphantom{g_1 = }{}\times
 \Gamma_{-}\big(q^{-\rho}\big)\Gamma'_{-}\big({-}Q_1q^{-\rho}\big) \Gamma_{-}\big(Q_1Q_3q^{-\rho}\big)\Gamma'_{-}\big({-}Q_1Q_2Q_3q^{-\rho}\big). \label{g1-ctv}
\end{gather}
In much the same way, $Z_{2,\beta}/Z$ turns out to be expressed as
\begin{gather*}
 Z_{2,\beta}/Z = \big\langle 0\,|\,g_2\,|\,\tp{\beta}\big\rangle,
\end{gather*}
where
\begin{gather}
 g_2 = \Gamma_{+}\big({-}Q_1Q_2Q_3q^{-\rho}\big)\Gamma'_{+}\big(Q_2Q_3q^{-\rho}\big) \Gamma_{+}\big({-}Q_2q^{-\rho}\big)\Gamma'_{+}\big(q^{-\rho}\big)\notag\\
\hphantom{g_2 =}{}\times
 \exp\left(\sum_{k=1}^\infty\frac{(Q_1Q_2q)^k}{k(1-q^k)}V^{(-k)}_0 \right) q^{-K/2}. \label{g2-ctv}
\end{gather}

These expressions of the normalized amplitudes lead us to the following conclusion.

\begin{Theorem} $\tau_1(\bsx)$ and $\tau_2(\bsx)$ are tau functions of the KP hierarchy with the fermionic expression
\begin{gather*}
 \tau_1(\bsx) = \big\langle 0\,|\,\Gamma_{+}(\bsx)g_1\,|\,0\big\rangle, \qquad \tau_2(\bsx) = \big\langle 0\,|\,g_2\Gamma_{-}(\bsx)\,|\,0\big\rangle,
\end{gather*}
where $g_1$ and $g_2$ are the operators defined in~\eqref{g1-ctv} and~\eqref{g2-ctv}.
\end{Theorem}

\begin{Remark}
The fermionic expression of $\tau_2(\bsx)$ can be rewritten as
\begin{gather*}
 \tau_2(\bsx) = \big\langle 0\,|\,\Gamma_{+}(\bsx)\tp{g_2}\,|\,0\big\rangle,
\end{gather*}
where
\begin{gather*}
 \tp{g_2} = q^{-K/2} \exp\left(\sum_{k=1}^\infty\frac{(Q_1Q_2q)^k}{k(1-q^k)}V^{(-k)}_0 \right) \notag\\
\hphantom{\tp{g_2} =}{}\times \Gamma'_{-}\big(q^{-\rho}\big)\Gamma_{-}\big({-}Q_2q^{-\rho}\big) \Gamma'_{-}\big(Q_2Q_3q^{-\rho}\big)\Gamma_{-}\big({-}Q_1Q_2Q_3q^{-\rho}\big).
\end{gather*}
\end{Remark}

\subsection{Admissible basis and Kac--Schwarz operators}

Since all building blocks of $\tau_1(\bsx)$ and $\tau_2(\bsx)$ are now familiar, one can derive the associated admissible bases and $q$-dif\/ference Kac--Schwarz operators in almost the same way as we have done in the preceding sections. Actually, the computations are partly more complicated, and exhibit a new aspect that is characteristic of this case. Let us consider $\tau_1(\bsx)$ to illustrate the computations.

The generating operator $G$ of an admissible basis $\{\Phi_j(x)\}_{j=0}^\infty$ for $\tau_1(\bsx)$
takes such a form as
\begin{gather*}
 G = \prod_{i=1}^\infty\big(1 - Q_1Q_2q^{i-1+D}\big)\cdot\Phi(x), \label{G-ctv*}
\end{gather*}
where $\Phi(x)$ is the composite dilogarithmic function
\begin{gather*}
 \Phi(x) = \prod_{i=1}^\infty \frac{(1 - Q_1q^{i-1/2}x)(1 - Q_1Q_2Q_3q^{i-1/2}x)} {(1 - q^{i-1/2}x)(1 - Q_1q_3q^{i-1/2}x)}.
\end{gather*}
The admissible basis is generated from $\Phi(x)$ as
\begin{gather*}
 \Phi_j(x) = \prod_{i=1}^\infty\big(1 - Q_1Q_2q^{i-1+D}\big)\big(x^{-j}\Phi(x)\big).
\end{gather*}

$\Phi_j(x)$'s are neither quantum dilogarithmic nor $q$-hypergeometric. Power series expansion of these functions take such a form as
\begin{gather}
 \Phi_j(x) = \prod_{i=1}^\infty\big(1 - Q_1Q_2q^{i-1-j}\big) \cdot\sum_{k=0}^\infty\frac{b_k}{(Q_1Q_2q^{-j};q)_k}x^{k-j}, \label{Phij-ctv}
\end{gather}
where $b_k$'s are coef\/f\/icients of the power series expansion
\begin{gather*}
 \Phi(x) = \sum_{k=0}^\infty b_kx^k
\end{gather*}
of $\Phi(x)$. Comparing the expansion of $\Phi_0(x)$ with the expansion of the generating function $\Psi_1(x)$ presented in our previous work \cite{TN15},
one can conf\/irm that $\Phi_0(x)$ agrees with $\Psi_1(x)$ except for the constant multiplier $\prod\limits_{i=1}^\infty(1 - Q_1Q_2q^{i-1})$.

In the previous work, we derived a $q$-dif\/ference equation for $\Psi_1(x)$ from the $q$-dif\/ference equation
\begin{gather}
 \Phi(qx) = \frac{(1 - q^{1/2}x)(1 - Q_1Q_3q^{1/2}x)} {(1 - Q_1q^{1/2}x)(1 - Q_1Q_2Q_3q^{1/2}x)}\Phi(x) \label{Phi-eq}
\end{gather}
for $\Phi(x)$ by a quite primitive method. The notion of Kac--Schwarz operators leads us to a more systematic derivation of the $q$-dif\/ference equation.

Let us proceed to the $q$-dif\/ference Kac--Schwarz operators. It is easy to compute $B = G\cdot x^{-1}\cdot G^{-1}$:
\begin{gather}
 B = \prod_{i=1}^\infty\big(1 - Q_1Q_2q^{i-1+D}\big)\cdot x^{-1} \cdot\prod_{i=1}^\infty\big(1 - Q_1Q_2q^{i-1+D}\big)^{-1} \notag\\
 = \prod_{i=1}^\infty\big(1 - Q_1Q_2q^{i-1+D}\big) \cdot\prod_{i=1}^\infty\big(1 - Q_1Q_2q^{i+D}\big)^{-1}\cdot x^{-1}
 = \big(1 - Q_1Q_2q^D\big)\cdot x^{-1}. \label{B-ctv}
\end{gather}
To compute $A = G\cdot q^{-D}\cdot G^{-1}$, one can start from the algebraic relation
\begin{gather*}
 \Phi(x)\cdot q^{-D}\cdot \Phi(x)^{-1} = q^{-D}\frac{(1 - q^{1/2}x)(1 - Q_1Q_3q^{1/2}x)} {(1 - Q_1q^{1/2}x)(1 - Q_1Q_2Q_3q^{1/2}x)},
\end{gather*}
which is the $q$-dif\/ference equation (\ref{Phi-eq}) in disguise. This relation yields the following intermediate expression of $A$:
\begin{gather*}
 A = q^{-D}\prod_{i=1}^\infty\big(1 - Q_1Q_2q^{i-1+D}\big)\cdot \frac{(1 - q^{1/2}x)(1 - Q_1Q_3q^{1/2}x)} {(1 - Q_1q^{1/2}x)(1 - Q_1Q_2Q_3q^{1/2}x)}\\
 \hphantom{A=}{}\times  \prod_{i=1}^\infty\big(1 - Q_1Q_2q^{i-1+D}\big)^{-1}.
\end{gather*}
Since (\ref{B-ctv}) implies that
\begin{gather*}
 \prod_{i=1}^\infty\big(1 - Q_1Q_2q^{i-1+D}\big)\cdot x \cdot\prod_{i=1}^\infty\big(1 - Q_1Q_2q^{i-1+D}\big)^{-1} = x\big(1 - Q_1Q_2q^D\big)^{-1},
\end{gather*}
one can rewrite the foregoing intermediate expression of $A$ as
\begin{gather}
 A = q^{-D} \frac{\left(1 - q^{1/2}x(1-Q_1Q_2q^D)^{-1}\right) \left(1 - Q_1Q_3q^{1/2}x(1-Q_1Q_2q^D)^{-1}\right)}
 {\left(1 - Q_1q^{1/2}x(1-Q_1Q_2q^D)^{-1}\right) \left(1 - Q_1Q_2Q_3q^{1/2}x(1-Q_1Q_2q^D)^{-1}\right)}. \label{A-ctv}
\end{gather}

Let us summarize these computations. Although we omit details, similar results hold for~$\tau_2(\bsx)$ as well.

\begin{Theorem} \eqref{B-ctv} and \eqref{A-ctv} are $q$-difference Kac--Schwarz operators for~$\tau_1(\bsx)$. \eqref{Phij-ctv} give an admissible basis
of the associated point of the Sato Grassmannian.
\end{Theorem}

\subsection[Remarks on $q$-dif\/ference equations]{Remarks on $\boldsymbol{q}$-dif\/ference equations}

Let us rewrite the equation $A\Phi_0(x) = \Phi_0(x)$ for $\Phi_0(x)$ as
\begin{gather}
 \big(1 - q^{1/2}x\big(1-Q_1Q_2q^D\big)^{-1}\big) \big(1 - Q_1Q_3q^{1/2}x\big(1-Q_1Q_2q^D\big)^{-1}\big)\Phi_0(x) \notag\\
\qquad {}= \big(1 - Q_1q^{1/2}x\big(1-Q_1Q_2q^D\big)^{-1}\big)
 \big(1 - Q_1Q_2Q_3q^{1/2}x\big(1-Q_1Q_2q^D\big)^{-1}\big)q^D\Phi_0(x).\!\!\!  \label{Phi0-eq-ctv}
\end{gather}
Although this equation looks considerably dif\/ferent from the $q$-dif\/ference equation presented in our previous work \cite{TN15}, they are equivalent
as we show below.

\begin{Lemma} For any constants $P_1$, $P_2$ and $Q$,
\begin{gather*}
 \big(1 - P_1x\big(1-Qq^D\big)^{-1}\big) \big(1 - P_2x\big(1-Qq^D\big)^{-1}\big) \notag\\
\qquad{}= \big(1 - Qq^{D-1}\big)^{-1}\big(1 - Qq^{D-2}\big)^{-1} \big(1 - Qq^{D-2} - P_1x\big)\big(1 - Qq^{D-1} - P_2x\big).
\end{gather*}
\end{Lemma}

\begin{proof}
Do straightforward computations using (\ref{x-conj}) as follows:
\begin{gather*}
 \big(1 - P_1x\big(1-Qq^D\big)^{-1}\big) \big(1 - P_2x\big(1-Qq^D\big)^{-1}\big) \\
\qquad{} = \big(1 - P_1\big(1-Qq^{D-1}\big)^{-1}x\big) \big(1 - P_2\big(1-Qq^{D-1}\big)^{-1}x\big)\\
 \qquad{} = \big(1 - Qq^{D-1}\big)^{-1}\big(1 - Qq^{D-1} - P_1x\big) \big(1 - Qq^{D-1}\big)^{-1}\big(1 - Qq^{D-1} - P_2x\big)\\
\qquad{} = \big(1 - Qq^{D-1}\big)^{-1}\big(1 - P_1x\big(1-Qq^{D-1}\big)^{-1}\big) \big(1 - Qq^{D-1} - P_2x\big)\\
 \qquad{} = \big(1 - Qq^{D-1}\big)^{-1}\big(1 - P_1\big(1-Qq^{D-2}x\big)^{-1}x\big) \big(1 - Qq^{D-1} - P_2x\big)\\
\qquad{} = \big(1 - Qq^{D-1}\big)^{-1}\big(1-Qq^{D-2}x\big)^{-1} \big(1 - Qq^{D-2} - P_1x\big) \big(1 - Qq^{D-1} - P_2x\big).\tag*{\qed}
\end{gather*}\renewcommand{\qed}{}
\end{proof}

One can use this lemma to rewrite both sides of (\ref{Phi0-eq-ctv}) and collect terms to the left side. The outcome reads
\begin{gather}
 \big(1 - Q_1Q_2q^{D-1}\big)^{-1}\big(1 - Q_1Q_2q^{D-2}\big)^{-1}\hat{H}\Phi_0(x) = 0, \label{Phi0-eq2-ctv}
\end{gather}
where
\begin{gather*}
 \hat{H}= \big(1 - Q_1Q_2q^{D-2} - q^{1/2}x\big)\big(1 - Q_1Q_2q^{D-1} - Q_1Q_3q^{1/2}x\big) \notag\\
\hphantom{\hat{H}=}{} - \big(1 - Q_1Q_2q^{D-2} - Q_1q^{1/2}x\big)\big(1 - Q_1Q_2q^{D-1} - Q_1Q_2Q_3q^{1/2}x\big)q^D.
\end{gather*}

We are thus substantially in the same situation as encountered in the previous work. Namely, as far as $Q_1Q_2$ is not equal to integral powers of~$q$ (and this condition is assumed implicitly in the foregoing consideration), the two operators in front of~$\hat{H}$ in~(\ref{Phi0-eq2-ctv}) are invertible on $V = \CC((x))$. Therefore one can remove these operators and obtain the equation
\begin{gather}
 \hat{H}\Phi_0(x) = 0,  \label{Phi0-eq3-ctv}
\end{gather}
which is one of the equivalent forms of the $q$-dif\/ference equations derived in the previous work. Moreover, as pointed out therein, $\hat{H}$ itself can be factorized as
\begin{gather*}
 \hat{H} = \big(1 - Q_1Q_2q^{D-2}\big)\hat{K},
\end{gather*}
where $\hat{K}$ is another $q$-dif\/ference operator. Therefore (\ref{Phi0-eq3-ctv}) can be further reduced to
\begin{gather}
 \hat{K}\Phi_0(x) = 0. \label{Phi0-eq4-ctv}
\end{gather}

It is this equation (\ref{Phi0-eq4-ctv}) that is proposed in our previous work as a candidate of the quantum mirror curve of close topological vertex. Actually, the foregoing computations show that~$\hat{K}$ is related to the operators on both sides of~(\ref{Phi0-eq-ctv}) as
\begin{gather*}
 \big(1 - q^{1/2}x\big(1-Q_1Q_2q^D\big)^{-1}\big) \big(1 - Q_1Q_3q^{1/2}x\big(1-Q_1Q_2q^D\big)^{-1}\big) \\
\qquad\quad{} - \big(1 - Q_1q^{1/2}x\big(1-Q_1Q_2q^D\big)^{-1}\big) \big(1 - Q_1Q_2Q_3q^{1/2}x\big(1-Q_1Q_2q^D\big)^{-1}\big)q^D \\
 \qquad{}= \big(1 -Q_1Q_2q^{D-1}\big)^{-1}\hat{K}.
\end{gather*}

\subsection*{Acknowledgements}

The authors are grateful to Motohico Mulase for discussion and encouragement. We owe him the idea that an integrable hierarchy may be thought of as a mirror map. This work is partly supported by JSPS Kakenhi Grant No.~25400111 and No.~15K04912.

\pdfbookmark[1]{References}{ref}
\LastPageEnding


\begin{thebibliography}{99}
\footnotesize\itemsep=0pt

\bibitem{ADKMV03}
Aganagic M., Dijkgraaf R., Klemm A., Mari\~no M., Vafa C., Topological strings
 and integrable hierarchies, \href{https://doi.org/10.1007/s00220-005-1448-9}{\textit{Comm. Math. Phys.}} \textbf{261} (2006),
 451--516, \href{http://arxiv.org/abs/hep-th/0312085}{hep-th/0312085}.

\bibitem{AKMV03}
Aganagic M., Klemm A., Mari\~no M., Vafa C., The topological vertex,
 \href{https://doi.org/10.1007/s00220-004-1162-z}{\textit{Comm. Math. Phys.}} \textbf{254} (2005), 425--478,
 \href{http://arxiv.org/abs/hep-th/0305132}{hep-th/0305132}.

\bibitem{Alexandrov1404}
Alexandrov A., Enumerative geometry, tau-functions and {H}eisenberg--{V}irasoro
 algebra, \href{https://doi.org/10.1007/s00220-015-2379-8}{\textit{Comm. Math. Phys.}} \textbf{338} (2015), 195--249,
 \href{http://arxiv.org/abs/1404.3402}{arXiv:1404.3402}.

\bibitem{ALS1512}
Alexandrov A., Lewanski D., Shadrin S., Ramif\/ications of {H}urwitz theory, {KP}
 integrability and quantum curves, \href{https://doi.org/10.1007/JHEP05(2016)124}{\textit{J.~High Energy Phys.}} \textbf{2016}
 (2016), no.~5, 124, 31~pages, \href{http://arxiv.org/abs/1512.07026}{arXiv:1512.07026}.

\bibitem{AMMN14}
Alexandrov A., Mironov A., Morozov A., Natanzon S., On {KP}-integrable
 {H}urwitz functions, \href{https://doi.org/10.1007/JHEP11(2014)080}{\textit{J.~High Energy Phys.}} \textbf{2014} (2014), no.~11, 080, 31~pages,
 \href{http://arxiv.org/abs/1405.1395}{arXiv:1405.1395}.

\bibitem{BTZ11}
Bonelli G., Tanzini A., Zhao J., Vertices, vortices \& interacting surface
 operators, \href{https://doi.org/10.1007/JHEP06(2012)178}{\textit{J.~High Energy Phys.}} \textbf{2012} (2012), no.~6, 178, 22~pages,
 \href{http://arxiv.org/abs/1102.0184}{arXiv:1102.0184}.

\bibitem{BKMS07}
Bouchard V., Klemm A., Mari\~no M., Pasquetti S., Remodeling the {B}-model,
 \href{https://doi.org/10.1007/s00220-008-0620-4}{\textit{Comm. Math. Phys.}} \textbf{287} (2009), 117--178, \href{http://arxiv.org/abs/0709.1453}{arXiv:0709.1453}.

\bibitem{BM07}
Bouchard V., Mari\~no M., Hurwitz numbers, matrix models and enumerative
 geometry, in From {H}odge Theory to Integrability and {TQFT} tt*-Geometry,
 \href{https://doi.org/10.1090/pspum/078/2483754}{\textit{Proc. Sympos. Pure Math.}}, Vol.~78, Amer. Math. Soc., Providence, RI,
 2008, 263--283, \href{http://arxiv.org/abs/0709.1458}{arXiv:0709.1458}.

\bibitem{BK03}
Bryan J., Karp D., The closed topological vertex via the {C}remona transform,
 \href{https://doi.org/10.1090/S1056-3911-04-00394-7}{\textit{J.~Algebraic Geom.}} \textbf{14} (2005), 529--542,
 \href{http://arxiv.org/abs/math.AG/0311208}{math.AG/0311208}.

\bibitem{DHS09}
Dijkgraaf R., Hollands L., Su{\l}kowski P., Quantum curves and {$\mathcal
 D$}-modules, \href{https://doi.org/10.1088/1126-6708/2009/11/047}{\textit{J.~High Energy Phys.}} \textbf{2009} (2009), no.~11, 047,
 59~pages, \href{http://arxiv.org/abs/0810.4157}{arXiv:0810.4157}.

\bibitem{DHSV08}
Dijkgraaf R., Hollands L., Su{\l}kowski P., Vafa C., Supersymmetric gauge
 theories, intersecting branes and free fermions, \href{https://doi.org/10.1088/1126-6708/2008/02/106}{\textit{J.~High Energy
 Phys.}} \textbf{2008} (2008), no.~2, 106, 57~pages, \href{http://arxiv.org/abs/0709.4446}{arXiv:0709.4446}.

\bibitem{DV02}
Dijkgraaf R., Vafa C., Matrix models, topological strings, and supersymmetric
 gauge theories, \href{https://doi.org/10.1016/S0550-3213(02)00766-6}{\textit{Nuclear Phys.~B}} \textbf{644} (2002), 3--20,
 \href{http://arxiv.org/abs/hep-th/0206255}{hep-th/0206255}.

\bibitem{DVV91}
Dijkgraaf R., Verlinde H., Verlinde E., Loop equations and {V}irasoro
 constraints in nonperturbative two-dimensional quantum gravity,
 \href{https://doi.org/10.1016/0550-3213(91)90199-8}{\textit{Nuclear Phys.~B}} \textbf{348} (1991), 435--456.

\bibitem{DGH10}
Dimofte T., Gukov S., Hollands L., Vortex counting and {L}agrangian
 3-manifolds, \href{https://doi.org/10.1007/s11005-011-0531-8}{\textit{Lett. Math. Phys.}} \textbf{98} (2011), 225--287,
 \href{http://arxiv.org/abs/1006.0977}{arXiv:1006.0977}.

\bibitem{Douglas90}
Douglas M.R., Strings in less than one dimension and the generalized {K}d{V}
 hierarchies, \href{https://doi.org/10.1016/0370-2693(90)91716-O}{\textit{Phys. Lett.~B}} \textbf{238} (1990), 176--180.

\bibitem{EK03}
Eguchi T., Kanno H., Geometric transitions, {C}hern--{S}imons gauge theory and
 {V}eneziano type amplitudes, \href{https://doi.org/10.1016/j.physletb.2004.01.085}{\textit{Phys. Lett.~B}} \textbf{585} (2004),
 163--172, \href{http://arxiv.org/abs/hep-th/0312223}{hep-th/0312223}.

\bibitem{EMS09}
Eynard B., Mulase M., Safnuk B., The {L}aplace transform of the cut-and-join
 equation and the {B}ouchard--{M}ari\~no conjecture on {H}urwitz numbers,
 \href{https://doi.org/10.2977/PRIMS/47}{\textit{Publ. Res. Inst. Math. Sci.}} \textbf{47} (2011), 629--670,
 \href{http://arxiv.org/abs/0907.5224}{arXiv:0907.5224}.

\bibitem{EO07}
Eynard B., Orantin N., Invariants of algebraic curves and topological
 expansion, \href{https://doi.org/10.4310/CNTP.2007.v1.n2.a4}{\textit{Commun. Number Theory Phys.}} \textbf{1} (2007), 347--452,
 \href{http://arxiv.org/abs/math-ph/0702045}{math-ph/0702045}.

\bibitem{FK93}
Faddeev L.D., Kashaev R.M., Quantum dilogarithm, \href{https://doi.org/10.1142/S0217732394000447}{\textit{Modern Phys. Lett.~A}}
 \textbf{9} (1994), 427--434, \mbox{\href{http://arxiv.org/abs/hep-th/9310070}{hep-th/9310070}}.

\bibitem{FV93}
Faddeev L.D., Volkov A.Yu., Abelian current algebra and the {V}irasoro algebra on
 the lattice, \href{https://doi.org/10.1016/0370-2693(93)91618-W}{\textit{Phys. Lett.~B}} \textbf{315} (1993), 311--318,
 \href{http://arxiv.org/abs/hep-th/9307048}{hep-th/9307048}.

\bibitem{FLZ1310}
Fang B., Liu C.-C.M., Zong Z., All genus mirror symmetry for toric
 {C}alabi--{Y}au 3-orbifolds, in String-{M}ath 2014, \textit{Proc. Sympos.
 Pure Math.}, Vol.~93, Amer. Math. Soc., Providence, RI, 2016, 1--19,
 \href{http://arxiv.org/abs/1310.4818}{arXiv:1310.4818}.

\bibitem{FKN91}
Fukuma M., Kawai H., Nakayama R., Inf\/inite-dimensional {G}rassmannian structure
 of two-dimensional quantum gravity, \href{https://doi.org/10.1007/BF02099014}{\textit{Comm. Math. Phys.}} \textbf{143}
 (1992), 371--403.

\bibitem{GR-book}
Gasper G., Rahman M., Basic hypergeometric series, \href{https://doi.org/10.1017/CBO9780511526251}{\textit{Encyclopedia of
 Mathematics and its Applications}}, Vol.~96, 2nd~ed., Cambridge University
 Press, Cambridge, 2004.

\bibitem{GPH1408}
Guay-Paquet M., Harnad J., Generating functions for weighted Hurwitz numbers,
 \href{http://arxiv.org/abs/1408.6766}{arXiv:1408.6766}.

\bibitem{GPH1405}
Guay-Paquet M., Harnad J., 2{D} {T}oda {$\tau$}-functions as combinatorial
 generating functions, \href{https://doi.org/10.1007/s11005-015-0756-z}{\textit{Lett. Math. Phys.}} \textbf{105} (2015),
 827--852, \href{http://arxiv.org/abs/1405.6303}{arXiv:1405.6303}.

\bibitem{GS11}
Gukov S., Su{\l}kowski P., A-polynomial, {B}-model, and quantization,
 \href{https://doi.org/10.1007/JHEP02(2012)070}{\textit{J.~High Energy Phys.}} \textbf{2008} (2012), no.~2, 070, 57~pages,
 \href{http://arxiv.org/abs/1108.0002}{arXiv:1108.0002}.

\bibitem{Harnad1410}
Harnad J., Multi-species quantum Hurwitz numbers, \href{http://arxiv.org/abs/1410.8817}{arXiv:1410.8817}.

\bibitem{Harnad1504}
Harnad J., Quantum {H}urwitz numbers and {M}acdonald polynomials,
 \href{https://doi.org/10.1063/1.4967953}{\textit{J.~Math. Phys.}} \textbf{57} (2016), 113505, 16~pages,
 \href{http://arxiv.org/abs/1504.03311}{arXiv:1504.03311}.

\bibitem{Harnad1504review}
Harnad J., Weighted {H}urwitz numbers and hypergeometric {$\tau$}-functions: an
 overview, in String-{M}ath 2014, \textit{Proc. Sympos. Pure Math.}, Vol.~93,
 Amer. Math. Soc., Providence, RI, 2016, 289--333, \href{http://arxiv.org/abs/1504.03408}{arXiv:1504.03408}.

\bibitem{HO1407}
Harnad J., Orlov A.Yu., Hypergeometric {$\tau$}-functions, {H}urwitz numbers and
 enumeration of paths, \href{https://doi.org/10.1007/s00220-015-2329-5}{\textit{Comm. Math. Phys.}} \textbf{338} (2015),
 267--284, \href{http://arxiv.org/abs/1407.7800}{arXiv:1407.7800}.

\bibitem{Hollands09}
Hollands L., Topological strings and quantum curves, Ph.D.~Thesis, University
 of Amsterdam, 2009, \href{http://arxiv.org/abs/0911.3413}{arXiv:0911.3413}.

\bibitem{HY06}
Hyun S., Yi S.-H., Non-compact topological branes on conifold, \href{https://doi.org/10.1088/1126-6708/2006/11/075}{\textit{J.~High
 Energy Phys.}} \textbf{2006} (2006), no.~11, 075, 34~pages,
 \href{http://arxiv.org/abs/hep-th/0609037}{hep-th/0609037}.

\bibitem{IKP04}
Iqbal A., Kashani-Poor A.-K., The vertex on a strip, \href{https://doi.org/10.4310/ATMP.2006.v10.n3.a2}{\textit{Adv. Theor. Math.
 Phys.}} \textbf{10} (2006), 317--343, \mbox{\href{http://arxiv.org/abs/hep-th/0410174}{hep-th/0410174}}.

\bibitem{KS91}
Kac V., Schwarz A., Geometric interpretation of the partition function of
 {$2$}{D} gravity, \href{https://doi.org/10.1016/0370-2693(91)91901-7}{\textit{Phys. Lett.~B}} \textbf{257} (1991), 329--334.

\bibitem{KP06}
Kashani-Poor A.-K., The wave function behavior of the open topological string
 partition function on the conifold, \href{https://doi.org/10.1088/1126-6708/2007/04/004}{\textit{J.~High Energy Phys.}}
 \textbf{2007} (2007), 004, 47~pages, \href{http://arxiv.org/abs/hep-th/0606112}{hep-th/0606112}.

\bibitem{KP08}
Kashani-Poor A.-K., Phase space polarization and the topological string: a case
 study, \href{https://doi.org/10.1142/S0217732308028958}{\textit{Modern Phys. Lett.~A}} \textbf{23} (2008), 3199--3214,
 \href{http://arxiv.org/abs/0812.0687}{arXiv:0812.0687}.

\bibitem{KMM95}
Kharchev S., Marshakov A., Mironov A., Morozov A., Generalized
 {K}azakov--{M}igdal--{K}ontsevich model: group theory aspects,
 \href{https://doi.org/10.1142/S0217751X9500098X}{\textit{Internat.~J. Modern Phys.~A}} \textbf{10} (1995), 2015--2051,
 \href{http://arxiv.org/abs/hep-th/9312210}{hep-th/9312210}.

\bibitem{Kontsevich92}
Kontsevich M., Intersection theory on the moduli space of curves and the matrix
 {A}iry function, \href{https://doi.org/10.1007/BF02099526}{\textit{Comm. Math. Phys.}} \textbf{147} (1992), 1--23.

\bibitem{KPW10}
Koz\c{c}az C., Pasquetti S., Wyllard N., A \& {B} model approaches to surface
 operators and {T}oda theories, \href{https://doi.org/10.1007/JHEP08(2010)042}{\textit{J.~High Energy Phys.}} \textbf{2010}
 (2010), no.~8, 042, 42~pages, \href{http://arxiv.org/abs/1004.2025}{arXiv:1004.2025}.

\bibitem{Mac-book}
Macdonald I.G., Symmetric functions and {H}all polynomials, 2nd ed., \textit{Oxford
 Mathematical Monographs}, The Clarendon Press, Oxford University Press, New
 York, 1995.

\bibitem{Marino06}
Mari\~no M., Open string amplitudes and large order behavior in topological
 string theory, \href{https://doi.org/10.1088/1126-6708/2008/03/060}{\textit{J.~High Energy Phys.}} \textbf{2008} (2008), no.~3, 060,
 34~pages, \href{http://arxiv.org/abs/hep-th/0612127}{hep-th/0612127}.

\bibitem{MJD-book}
Miwa T., Jimbo M., Date E., Solitons. Dif\/ferential equations, symmetries and
 inf\/inite-dimensional algebras, \textit{Cambridge Tracts in Mathematics}, Vol.~135, Cambridge University Press, Cambridge, 2000.

\bibitem{Moore90}
Moore G., Geometry of the string equations, \href{https://doi.org/10.1007/BF02097368}{\textit{Comm. Math. Phys.}}
 \textbf{133} (1990), 261--304.

\bibitem{MSS13}
Mulase M., Shadrin S., Spitz L., The spectral curve and the {S}chr\"odinger
 equation of double {H}urwitz numbers and higher spin structures,
 \href{https://doi.org/10.4310/CNTP.2013.v7.n1.a4}{\textit{Commun. Number Theory Phys.}} \textbf{7} (2013), 125--143,
 \href{http://arxiv.org/abs/1301.5580}{arXiv:1301.5580}.

\bibitem{MZ10}
Mulase M., Zhang N., Polynomial recursion formula for linear {H}odge integrals,
 \href{https://doi.org/10.4310/CNTP.2010.v4.n2.a1}{\textit{Commun. Number Theory Phys.}} \textbf{4} (2010), 267--293,
 \href{http://arxiv.org/abs/0908.2267}{arXiv:0908.2267}.

\bibitem{Nagao09}
Nagao K., Non-commutative {D}onaldson--{T}homas theory and vertex operators,
 \href{https://doi.org/10.2140/gt.2011.15.1509}{\textit{Geom. Topol.}} \textbf{15} (2011), 1509--1543, \href{http://arxiv.org/abs/0910.5477}{arXiv:0910.5477}.

\bibitem{NT07}
Nakatsu T., Takasaki K., Melting crystal, quantum torus and {T}oda hierarchy,
 \href{https://doi.org/10.1007/s00220-008-0583-5}{\textit{Comm. Math. Phys.}} \textbf{285} (2009), 445--468, \href{http://arxiv.org/abs/0710.5339}{arXiv:0710.5339}.

\bibitem{Okounkov00}
Okounkov A., Toda equations for {H}urwitz numbers, \href{https://doi.org/10.4310/MRL.2000.v7.n4.a10}{\textit{Math. Res. Lett.}}
 \textbf{7} (2000), 447--453, \href{http://arxiv.org/abs/math.AG/0004128}{math.AG/0004128}.

\bibitem{OP03}
Okounkov A., Pandharipande R., Gromov--{W}itten theory, {H}urwitz theory, and
 completed cycles, \href{https://doi.org/10.4007/annals.2006.163.517}{\textit{Ann. of Math.}} \textbf{163} (2006), 517--560,
 \href{http://arxiv.org/abs/math.AG/0204305}{math.AG/0204305}.

\bibitem{ORV03}
Okounkov A., Reshetikhin N., Vafa C., Quantum {C}alabi--{Y}au and classical
 crystals, in The Unity of Mathematics, \href{https://doi.org/10.1007/0-8176-4467-9_16}{\textit{Progr. Math.}}, Vol.~244,
 Editors P.~Etingof, V.~Retakh, I.M.~Singer, Birkh\"auser Boston, Boston, MA,
 2006, 597--618, \href{http://arxiv.org/abs/hep-th/0309208}{hep-th/0309208}.

\bibitem{OS01a}
Orlov A.Yu., Shcherbin D.M., Hypergeometric solutions of soliton equations,
 \href{https://doi.org/10.1023/A:1010402200567}{\textit{Theoret. and Math. Phys.}} \textbf{128} (2001), 906--926.

\bibitem{OS01b}
Orlov A.Yu., Scherbin D.M., Multivariate hypergeometric functions as
 {$\tau$}-functions of {T}oda lattice and {K}adomtsev--{P}etviashvili
 equation, \href{https://doi.org/10.1016/S0167-2789(01)00158-0}{\textit{Phys.~D}} \textbf{152/153} (2001), 51--65, \href{http://arxiv.org/abs/math-ph/0003011}{math-ph/0003011}.

\bibitem{SS83} Sato M., Sato Y., Soliton equations as dynamical systems on inf\/inite-dimensional {G}rassmann manifold, in Nonlinear Partial Differential Equations in Applied Science ({T}okyo, 1982), Editors~ H.~Fujita, P.D.~Lax, G.~Strang, \textit{North-Holland Math. Stud.}, Vol.~81,
North-Holland, Amsterdam, 1983, 259--271.

\bibitem{Schwarz91}
Schwarz A., On solutions to the string equation, \href{https://doi.org/10.1142/S0217732391003171}{\textit{Modern Phys. Lett.~A}}
 \textbf{6} (1991), 2713--2725.

\bibitem{Schwarz1401}
Schwarz A., Quantum curves, \href{https://doi.org/10.1007/s00220-015-2287-y}{\textit{Comm. Math. Phys.}} \textbf{338} (2015),
 483--500, \href{http://arxiv.org/abs/1401.1574}{arXiv:1401.1574}.

\bibitem{SW85}
Segal G., Wilson G., Loop groups and equations of {K}d{V} type, \href{https://doi.org/10.1007/BF02698802}{\textit{Inst.
 Hautes \'Etudes Sci. Publ. Math.}} \textbf{61} (1985), 5--65.

\bibitem{Sulkowski06}
Su{\l}kowski P., Crystal model for the closed topological vertex geometry,
 \href{https://doi.org/10.1088/1126-6708/2006/12/030}{\textit{J.~High Energy Phys.}} \textbf{2006} (2006), no.~12, 030, 21~pages,
 \href{http://arxiv.org/abs/hep-th/0606055}{hep-th/0606055}.

\bibitem{Sulkowski09}
Su{\l}kowski P., Wall-crossing, free fermions and crystal melting,
 \href{https://doi.org/10.1007/s00220-010-1153-1}{\textit{Comm. Math. Phys.}} \textbf{301} (2011), 517--562, \href{http://arxiv.org/abs/0910.5485}{arXiv:0910.5485}.

\bibitem{Takasaki10}
Takasaki K., Generalized string equations for double {H}urwitz numbers,
 \href{https://doi.org/10.1016/j.geomphys.2011.12.005}{\textit{J.~Geom. Phys.}} \textbf{62} (2012), 1135--1156, \href{http://arxiv.org/abs/1012.5554}{arXiv:1012.5554}.

\bibitem{Takasaki13}
Takasaki K., Remarks on partition functions of topological string theory on
 generalized conifolds, \href{http://arxiv.org/abs/1301.4548}{arXiv:1301.4548}.

\bibitem{TN15}
Takasaki K., Nakatsu T., Open string amplitudes of closed topological vertex,
 \href{https://doi.org/10.1088/1751-8113/49/2/025201}{\textit{J.~Phys.~A: Math. Theor.}} \textbf{49} (2016), 025201, 28~pages,
 \href{http://arxiv.org/abs/1507.07053}{arXiv:1507.07053}.

\bibitem{Taki10}
Taki M., Surface operator, bubbling {C}alabi--{Y}au and {AGT} relation,
 \href{https://doi.org/10.1007/JHEP07(2011)047}{\textit{J.~High Energy Phys.}} \textbf{2011} (2011), no.~7, 047, 33~pages,
 \href{http://arxiv.org/abs/1007.2524}{arXiv:1007.2524}.

\bibitem{Witten91}
Witten E., Two-dimensional gravity and intersection theory on moduli space, in
 Surveys in Dif\/ferential Geometry ({C}ambridge, {MA}, 1990), Lehigh
 University, Bethlehem, PA, 1991, 243--310.

\bibitem{BY08}
Young B., Generating functions for colored 3{D} {Y}oung diagrams and the
 {D}onaldson--{T}homas invariants of orbifolds (with an appendix by J.~Bryan), \href{https://doi.org/10.1215/00127094-2010-009}{\textit{Duke Math.~J.}}
 \textbf{152} (2010), 115--153, \href{http://arxiv.org/abs/0802.3948}{arXiv:0802.3948}.

\bibitem{Zabrodin12}
Zabrodin A., Laplacian growth in a channel and {H}urwitz numbers,
 \href{https://doi.org/10.1088/1751-8113/46/18/185203}{\textit{J.~Phys.~A: Math. Theor.}} \textbf{46} (2013), 185203, 23~pages,
 \href{http://arxiv.org/abs/1212.6729}{arXiv:1212.6729}.

\bibitem{Zhou1207}
Zhou J., Quantum mirror curves for $\mathbb{C}^3$ and the resolved conifold,
 \href{http://arxiv.org/abs/1207.0598}{arXiv:1207.0598}.

\bibitem{Zhou1512}
Zhou J., Emergent geometry of KP hierarchy,~II, \href{http://arxiv.org/abs/1512.03196}{arXiv:1512.03196}.

\end{thebibliography}
\end{document}